\documentclass[12pt]{amsart}
\usepackage[cp1251]{inputenc}
\usepackage[english]{babel}
\usepackage{amsaddr}
\usepackage{amsmath}
\usepackage{amssymb}
\usepackage{amsfonts}
\usepackage{epsfig}
\usepackage{srcltx}
\usepackage{subfigure}
\usepackage{cite}
\usepackage{float}
\usepackage[a4paper,  mag=1000, includefoot,  left=2cm, right=2cm, top=2cm, bottom=2cm, headsep=1cm, footskip=1cm ]{geometry}

\newtheorem{Th}{Theorem}

\newtheorem{Def}{Definition}
\newtheorem{Cor}{Corollary}
\begin{document}
\thispagestyle{empty}

{\small
\title[Autoresonance in systems with combined excitation and weak dissipation]
{Autoresonance in oscillating systems with combined excitation and weak dissipation}

\author{Oskar A. Sultanov}
\address{Institute of Mathematics, Ufa Federal Research Center, Russian Academy of Sciences,  112, Chernyshevsky str., Ufa 450008 Russia.}
\email{oasultanov@gmail.com}

\maketitle {\small
\begin{quote}
\noindent{\bf Abstract.}
A mathematical model describing the initial stage of the capture into autoresonance for nonlinear oscillating systems with combined parametric and external excitation is considered. The solutions with unboundedly growing amplitude and limited phase mismatch correspond to the autoresonant capture. The paper investigates the existence, stability and bifurcations of such solutions in the presence of a weak dissipation in the system. Our technique is based on the study of particular solutions with power-law asymptotics at infinity and the construction of suitable Lyapunov functions.
\medskip


%34C15  	Nonlinear oscillations, coupled oscillators 34D05  	Asymptotic properties 37B25  	Lyapunov functions and stability; attractors, repellers, 	37B55  	Nonautonomous dynamical systems 	93D20  	Asymptotic stability	
\end{quote}
}

\section*{Introduction}
Autoresonance is a phenomenon that occurs in nonlinear systems with slowly varying oscillating perturbations. Under certain conditions, the system automatically adjusts to the disturbances and holds this state for a sufficiently long period of time. As a result, the energy of the system can increase significantly~\cite{LFS09}. The autoresonance was first studied in the problems associated with the acceleration of particles~\cite{V44,M45} and planetary dynamics~\cite{G73,S78}. Nowadays, it is considered as a universal phenomenon with a wide range of applications~\cite{FGF00,USM10,Aetal11,KHM14,LF19etal,AGLF19,BShF18,FSh20}. The study of the corresponding mathematical models leads to new and challenging problems in the field of nonlinear dynamics~\cite{LFJPA08,LKRMS08,NVA13}.

Mathematical models associated with the autoresonance have been studied in many papers. See, for instance,~\cite{FGF01,LK03,LKMS09,GKT10,AK20}, where the systems with external driving were analyzed, and~\cite{KMPRE01,AM05,OS16,OK18,OS19}, where the models of parametric autoresonance were investigated. The effect of a combined external and parametric excitation on the autoresonant capture in nonlinear systems was first studied in~\cite{OS18,OS20}. In this paper, the autoresonance model with the combined excitation in the presence of a weak dissipation is considered, and the existence and stability of different autoresonant modes are discussed.

The paper is organized as follows. In section 1, the mathematical formulation of the problem is given. In section 2, the particular autoresonant solutions are described and the partition of a parameter space is constructed. The stability of particular solutions and asymptotics for general autoresonant solutions are discussed in section 3. A discussion of the results obtained is contained in section 4.

\section{Problem statement}
Consider the non-autonomous system of two differential equations:
\begin{gather}
    \label{MS}
    \begin{split}
        \frac{d\rho}{d\tau}+\gamma(\tau) \rho=\alpha(\tau)\sin\psi -\beta(\tau) \rho \sin (2\psi+\nu), \\
        \rho\Big[\frac{d\psi}{d\tau}-\rho^2+ \lambda \tau\Big]= \alpha(\tau)\cos\psi -\beta(\tau) \rho \cos (2\psi+\nu),
    \end{split}
\end{gather}
with the parameters $\lambda\neq 0$ and $\nu\in [0,\pi)$. Smooth given functions $\alpha(\tau)\not\equiv 0$ and $\beta(\tau)$ correspond to the amplitude of an external and a parametric driving, a positive function $\gamma(\tau)$ is associated with a dissipation. This system arises in the study of the autoresonance phenomena in a class of nonlinear oscillatory systems with a combined chirped-frequency excitation and a weak dissipation. The functions $\rho(\tau)$, $\psi(\tau)$ describe the evolution of the amplitude and the phase mismatch of the oscillators. The solutions with $\rho(\tau)\sim\sqrt{\lambda\tau}$ and $\psi(\tau)\sim\sigma$, $\sigma={\hbox{\rm const}}$ as $\tau\to \infty$ are associated with the phase-locking phenomenon and the capture into autoresonance. Note that system \eqref{MS} also has non-autoresonant solutions with a bounded amplitude, but such solutions are not considered in the present paper.

The combined effect of parametric and external excitations is determined by the behaviour of the ratio $f(\tau)\equiv \beta(\tau)/\alpha(\tau)$ as $\tau\to\infty$.  Indeed, if $f(\tau)\sim f_0 \tau^{-1/2-\varkappa}$, $\varkappa>0$, $f_0={\hbox{\rm const}}\neq 0$, the parametric pumping is insignificant and system \eqref{MS} corresponds to a perturbation of the model with the external driving. If $f(\tau)\sim f_0 \tau^{-1/2+\varkappa}$, the impact of external driving becomes inconsiderable and the system takes the form of a perturbed model of parametric autoresonance. The parametric and external excitations are comparable when $f(\tau)\sim f_0 \tau^{-1/2}$.
Note also that the existence of autoresonant solutions in systems with a dissipation depends on the behaviour of the function $g(\tau)\equiv\gamma(\tau)/\alpha(\tau)$ as $\tau\to\infty$. From the first equation in \eqref{MS} it follows that the necessary condition is $g(\tau)\sim g_0 \tau^{-1/2-\varkappa}$ with $\varkappa\geq 0$, $g_0={\hbox{\rm const}}\neq 0$.  Thus, in this paper it is assumed that
\begin{gather*}
    \alpha(\tau)=\tau^{\frac 12} \sum_{k=0}^\infty \alpha_k \tau^{-k}, \quad
    \beta(\tau)=\sum_{k=0}^\infty \beta_k \tau^{-k},\quad
    \gamma(\tau)=\sum_{k=0}^\infty \gamma_k \tau^{-k},\quad
        \tau\to\infty, \quad \alpha_k,\beta_k,\gamma_k={\hbox{\rm const}}.
\end{gather*}
Without loss of generality, we assume that $\alpha_0=1$ and $\gamma_0\neq 0$.

Note that system \eqref{MS} appears after averaging of perturbed oscillatory nonlinear systems and describes a long term evolution of solutions. For a system with one degree of freedom, the example is given by the following equation:
\begin{gather}
    \label{ex}
        \frac{d^2x}{dt^2}+\epsilon C(\epsilon t) \frac{dx}{dt}+\Big(1+\epsilon B(\epsilon t)\cos\big(2\zeta(t)-\nu\big) \Big) U'(x)=\epsilon A(\epsilon t)\cos\zeta(t),
\end{gather}
where $\zeta(t)=t-\vartheta t^2$, $U(x)=x^2/2-\epsilon x^4/24+\mathcal O(\epsilon^2)$, $0<\epsilon,\vartheta\ll 1$. We see that equation \eqref{ex} with $\epsilon=0$ has a stable trivial solution $x(t)\equiv 0$, $\dot x(t)\equiv 0$. Solutions of the perturbed equation with small enough initial data $(x(0),x'(0))$, whose the energy $E(t)\equiv U(x(t))+(x'(t))^2/2$ increases significantly  with time and the phase $\Phi(t)$ is synchronised with the pumping such that $\Phi(t)-\zeta(t)=\mathcal O(1)$, correspond to the capture into autoresonance. The approximation of such solutions is constructed by using the method of two scales with slow and fast variables: $\tau=\epsilon t/4$ and $\zeta=\zeta(t)$. The substitution
\begin{gather*}
  x(t)= 2 \rho(\tau) \cos \big(\zeta+\psi(\tau)\big) +\mathcal O(\epsilon)
\end{gather*}
into equation \eqref{ex} and the averaging over the fast variable lead to system \eqref{MS} for the slowly varying functions $\rho(\tau)$ and $\psi(\tau)$ with  $\lambda=16 \vartheta  \epsilon^{-2}$, $\alpha(\tau)=A(\epsilon t)$, $\beta(\tau)=B(\epsilon t)$, $\gamma(\tau)=2 C(\epsilon t)$. Likewise, system \eqref{MS} is derived in many other nonlinear problems related to autoresonance, including infinite-dimensional systems (see~\cite{LKRMS08}.

In this paper, the conditions for the existence and stability of autoresonant solutions to system \eqref{MS} are discussed. Our technique is based on the analysis of particular solutions with power-law asymptotics at infinity. In the first step, such solutions are constructed and the conditions for their existence specify the partition of the parameter space. Then, the Lyapunov stability of the particular solutions is investigated. Since the considered system is non-autonomous, the use of linear stability analysis is limited and nonlinear terms of equations must be taken into account. In this case, the stability can be justified with the Lyapunov function method. The presence of stability will ensure the existence of a family of autoresonant solutions. For such solutions, the asymptotic estimates at infinity are obtained at the last step from the properties of the constructed Lyapunov functions.

\section{Particular autoresonant solutions}
Consider the particular autoresonant solutions having the following asymptotics:
\begin{gather}
    \label{PAS}
        \rho_\ast(\tau)=\rho_{-1}\sqrt \tau+\rho_0+\sum_{k=1}^{\infty} \rho_k \tau^{-\frac k2}, \quad \psi_\ast(\tau)=\psi_0+\sum_{k=1}^\infty\psi_k\tau^{-\frac k2}, \quad \tau\to\infty.
\end{gather}
Substituting these series into system \eqref{MS} and grouping the terms of the same power of $\tau$ yield $\rho_{-1}= \sqrt\lambda$, $\rho_0=0$, and $\psi_0=\sigma$, where $\sigma$ satisfies the equation
\begin{gather}
    \label{teq}
    \mathcal P(\sigma;\delta,\nu,\kappa)\equiv \delta\sin (2\sigma+\nu)-\sin\sigma+\kappa=0,  \quad \delta=\beta_0 \sqrt \lambda, \quad \kappa=\gamma_0\sqrt \lambda.
\end{gather}
Note that the number of roots to equation  depends on the values of the parameters $(\delta,\nu,\kappa)$.
If, in addition, the inequality $\mathcal P'(\sigma;\delta,\nu,\kappa)\neq 0$ holds, the remaining coefficients $\rho_k$, $\psi_k$ as $k\geq 1$ are determined from the chain of linear equations:
\begin{gather*}
%    \label{sysFG}
        \begin{split}
            2\sqrt \lambda \rho_k & =\mathcal A_k(\rho_{-1},\dots,\rho_{k-1}, \sigma,\psi_1,\dots,\psi_{k-1}),  \\
            \mathcal P'(\sigma;\delta,\nu,\kappa)\psi_k & = \mathcal B_k(\rho_{-1},\dots,\rho_{k-1},\sigma, \psi_1,\dots,\psi_{k-1}),
    \end{split}
\end{gather*}
where
\begin{eqnarray*}
    \mathcal A_1&=&\frac{1}{\sqrt \lambda}(\delta \cos (2\sigma+\nu)-\cos\sigma),\\
    \mathcal A_2&=&-\frac{\psi_1}{\sqrt \lambda}(2\delta \sin (2\sigma+\nu)-\sin\sigma),\\
    \mathcal A_3 &=& -\rho_1^2 +\beta_1 \cos(2\sigma+\nu)-\Big(\alpha_1-\frac{\rho_1}{\sqrt\lambda}\Big)\frac{\cos\sigma}{\sqrt\lambda} \\&&-\frac{\psi_2}{\sqrt \lambda} (2 \delta \sin(2\sigma+\nu)-\sin\sigma)- \frac{\psi_1^2}{2\sqrt \lambda} (4\delta \cos(2\sigma+\nu)-\cos\sigma), \\
     \mathcal B_1&=&0,\\
     \mathcal B_2 &=&-\mathcal P''(\sigma;\delta,\nu,\kappa)\frac{\psi_1^2}{2}+\Big(\alpha_1-\frac{\rho_1}{\sqrt\lambda}\Big) \sin\sigma-\beta_1\sqrt\lambda\sin(2\sigma+\nu) -(1 + 2\gamma_1 )\frac{\rho_{-1}}{2}, \\
    \mathcal B_3&=& -\psi_1\psi_2 \mathcal P''(\sigma;\delta,\nu,\kappa) - \frac{\rho_2}{\sqrt\lambda}\sin\sigma - \frac{\psi_1^3}{6} \mathcal P'''(\sigma;\delta,\nu,\kappa)+ \alpha_1\psi_1\cos\sigma\\
        &&-2\psi_1 (\beta_1\sqrt\lambda+\beta_0\rho_1)\cos(2\sigma+\nu),
\end{eqnarray*}
etc. In particular,
\begin{gather*}
\psi_1=0,\quad  \psi_2=\theta, \quad \theta:=\frac{\mathcal B_2(\sigma)}{\mathcal P'(\sigma;\delta,\nu,\kappa)}.
\end{gather*}

Note that the pair of equations $\mathcal P(\sigma;\delta,\nu,\kappa)=0$ and $\mathcal P'(\sigma;\delta,\nu,\kappa)=0$ defines a bifurcation surface $S=S_1\cup S_2$ in the parameter space $(\delta,\nu,\kappa)$, where
\begin{gather*}
     S_j:=\{(\delta,\nu,\kappa)\in\mathbb R\times [0,\pi)\times \mathbb R: \sin \nu=p_j(\delta,\kappa)\},\\  p_j(\delta,\kappa):=\delta^{-1}\Big(\kappa (2\sin^2\varsigma_j-1)-\sin^3\varsigma_j\Big), \quad
  \sin\varsigma_j=z_j(\delta,\kappa), \\
  z_j(\delta,\kappa):=\frac{1}{3}\Big(4\kappa+(-1)^j\sqrt{4\kappa^2+12\delta^2-3}\Big), \quad j\in\{1,2\}.
\end{gather*}
For every $\kappa>0$, the bifurcation set is determined by the properties of $p_1(\delta,\kappa)$ and $p_2(\delta,\kappa)$ (see Fig.~\ref{fig1}).
\begin{figure}
\centering
\subfigure[$\kappa=0.4$]{\includegraphics[width=0.3\linewidth]{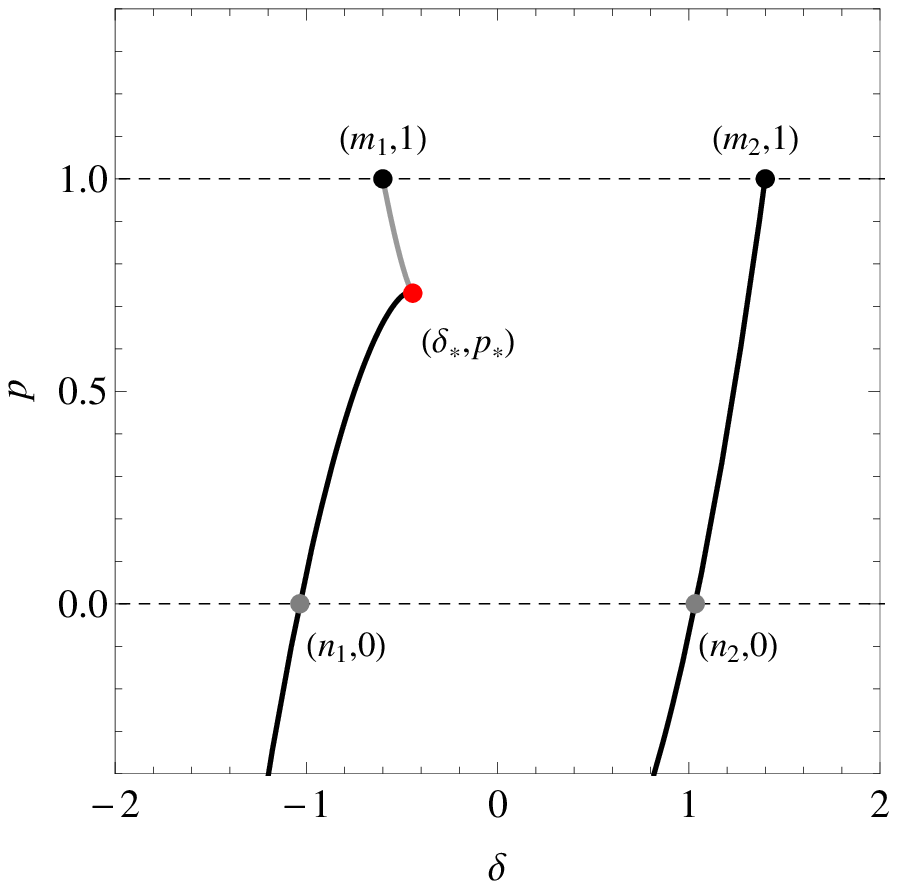}}
\hspace{2ex}
\subfigure[$\kappa=0.9$]{\includegraphics[width=0.3\linewidth]{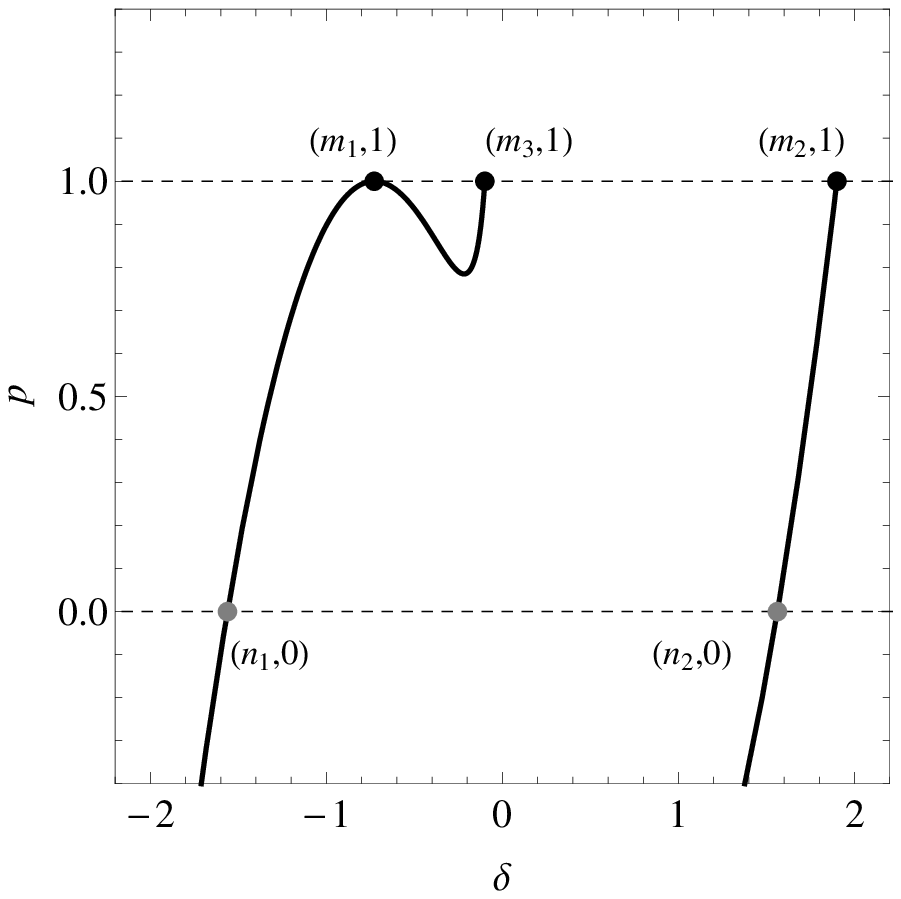}}
\\
\centering
\subfigure[$\kappa=1$]{\includegraphics[width=0.3\linewidth]{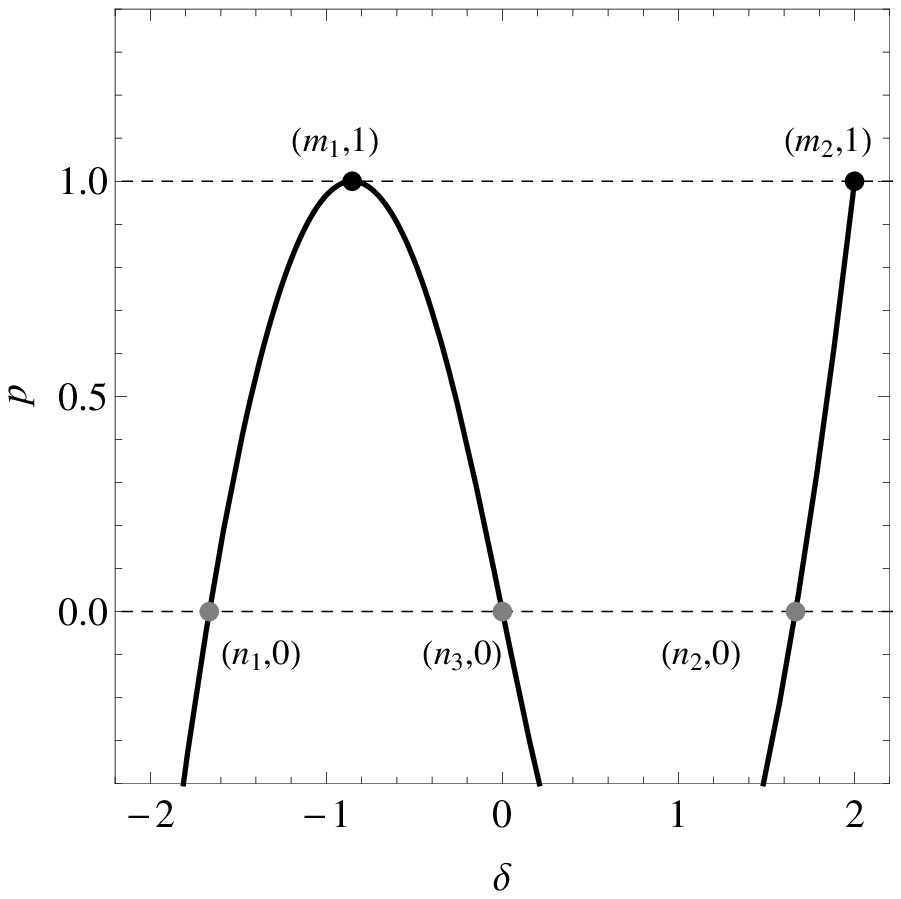}}
\hspace{2ex}
\subfigure[$\kappa=1.6$]{\includegraphics[width=0.3\linewidth]{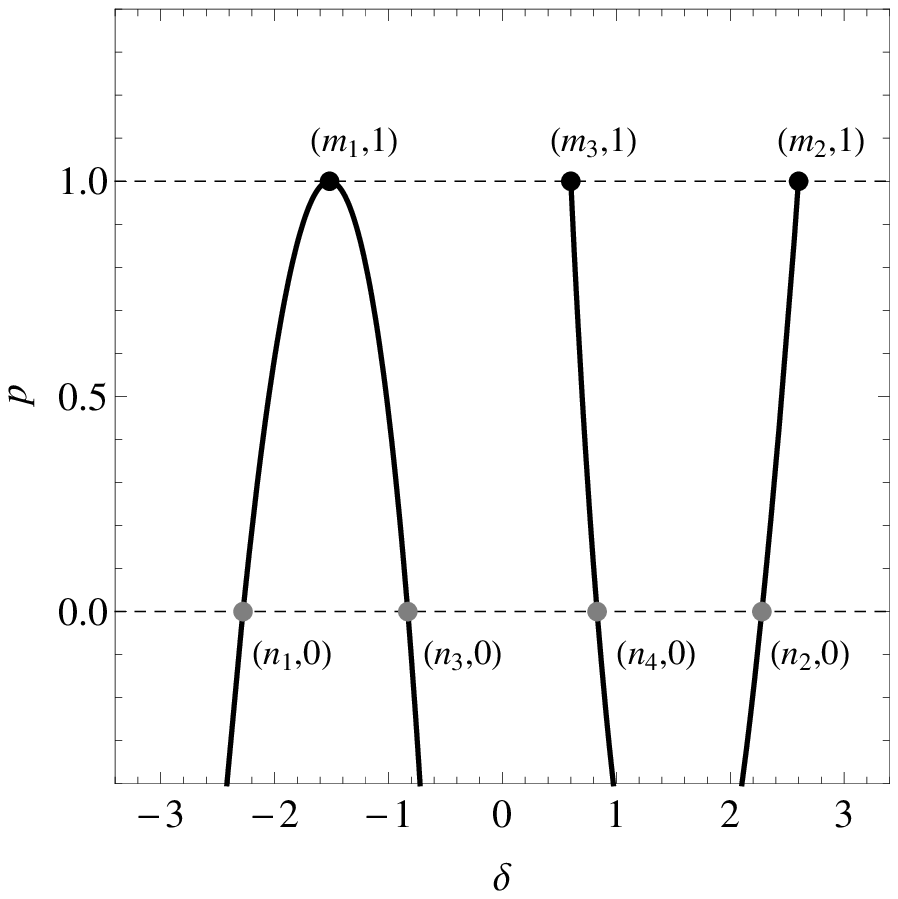}}
\caption{\small Graphs of $p_1(\delta,\kappa)$ (black curves) and $p_2(\delta,\kappa)$ (gray curves) as functions of the parameter $\delta$ with fixed $\kappa$. } \label{fig1}
\end{figure}
In particular, if $0<\kappa<3/4$, there are two curves
\begin{eqnarray*}
    s_+&:=&\{(\delta,\nu)\in [n_2,m_2]\times [0,\pi): \sin\nu=p_1(\delta,\kappa)\},\\
  s_- &:=&\{(\delta,\nu)\in [n_1,\delta_\ast]\times [0,\pi): \sin\nu=p_1(\delta,\kappa)\}\cup \{(\delta,\nu)\in [m_1,\delta_\ast]\times [0,\pi):\sin\nu=p_2(\delta,\kappa)\},
  \end{eqnarray*}
dividing the parameter plane $(\delta,\nu)$ into tree parts (see Fig.~\ref{fig2},a):
\begin{eqnarray*}
    \Omega_{+}&:{=}&\{(\delta,\nu)\in\mathbb R\times[0,\pi): \delta>s_+\},
     \quad
   \Omega_{-}:{=}\{(\delta,\nu)\in\mathbb R\times[0,\pi): \delta<s_-\},\\
    \Omega_{0}&:{=}&\{(\delta,\nu)\in\mathbb R\times[0,\pi): s_-< \delta<s_+\},
\end{eqnarray*}
where
$\delta_\ast=-\sqrt{(3-4\kappa^2)/12}$, $n_1<n_2$ are the roots of the equation $p_1(n,\kappa)=0$, $m_1=\kappa-1$, $m_2=\kappa+1$.
\begin{figure}
\centering
\subfigure[$\kappa=0.4$]{\includegraphics[width=0.3\linewidth]{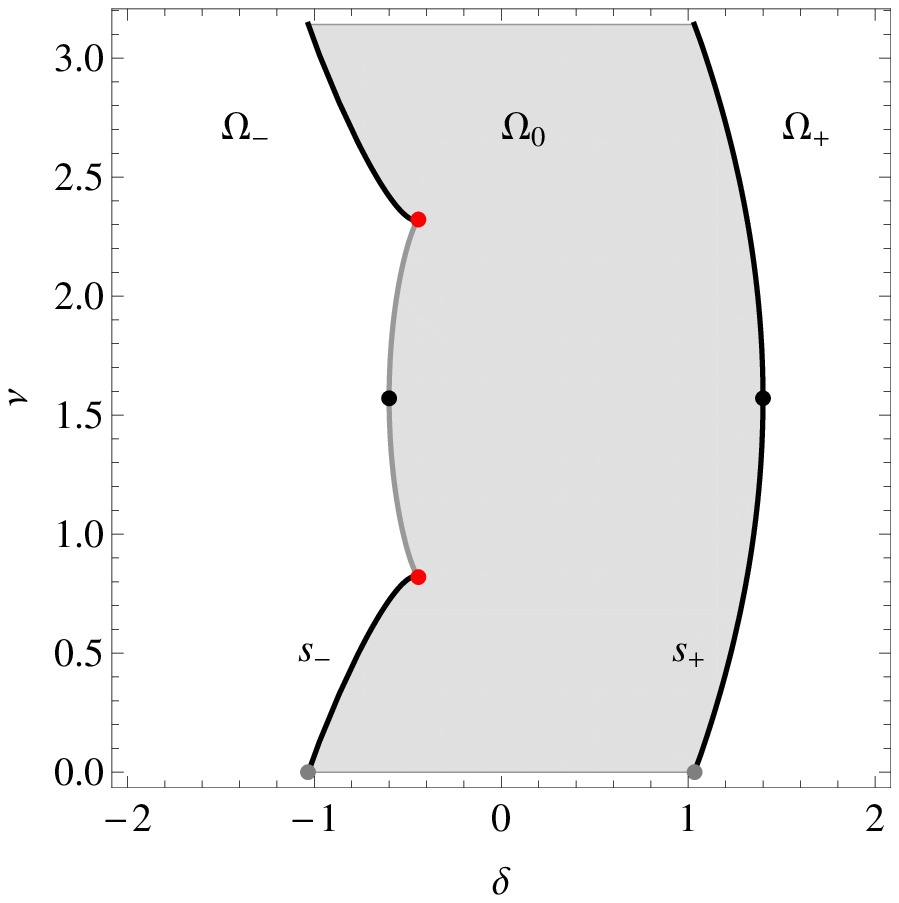}}
\hspace{2ex}
\subfigure[$\kappa=0.9$]{\includegraphics[width=0.3\linewidth]{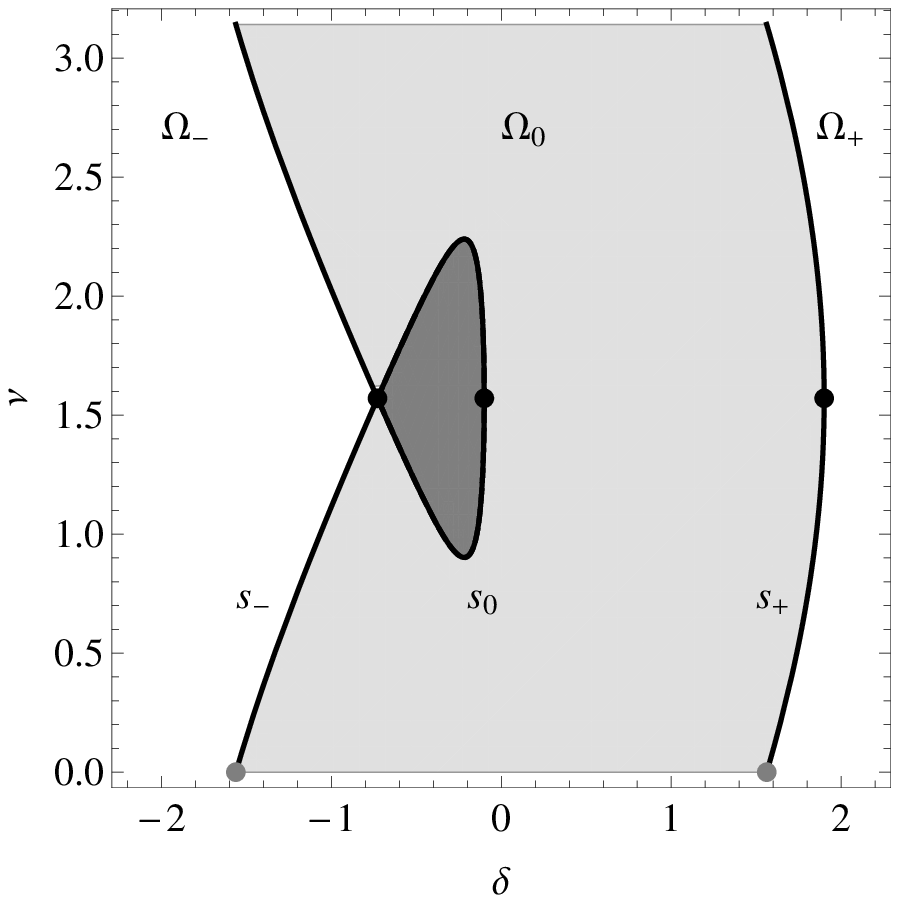}}
\\
\centering
\subfigure[$\kappa=1$]{\includegraphics[width=0.3\linewidth]{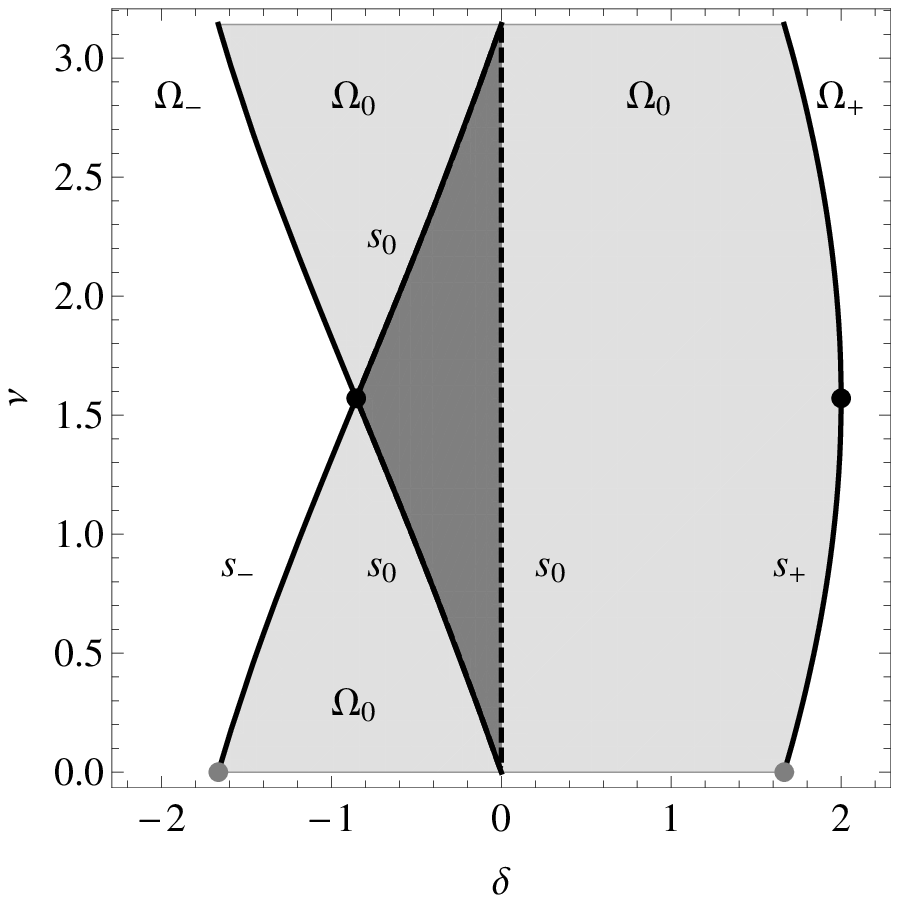}}
\hspace{2ex}
\subfigure[$\kappa=1.6$]{\includegraphics[width=0.3\linewidth]{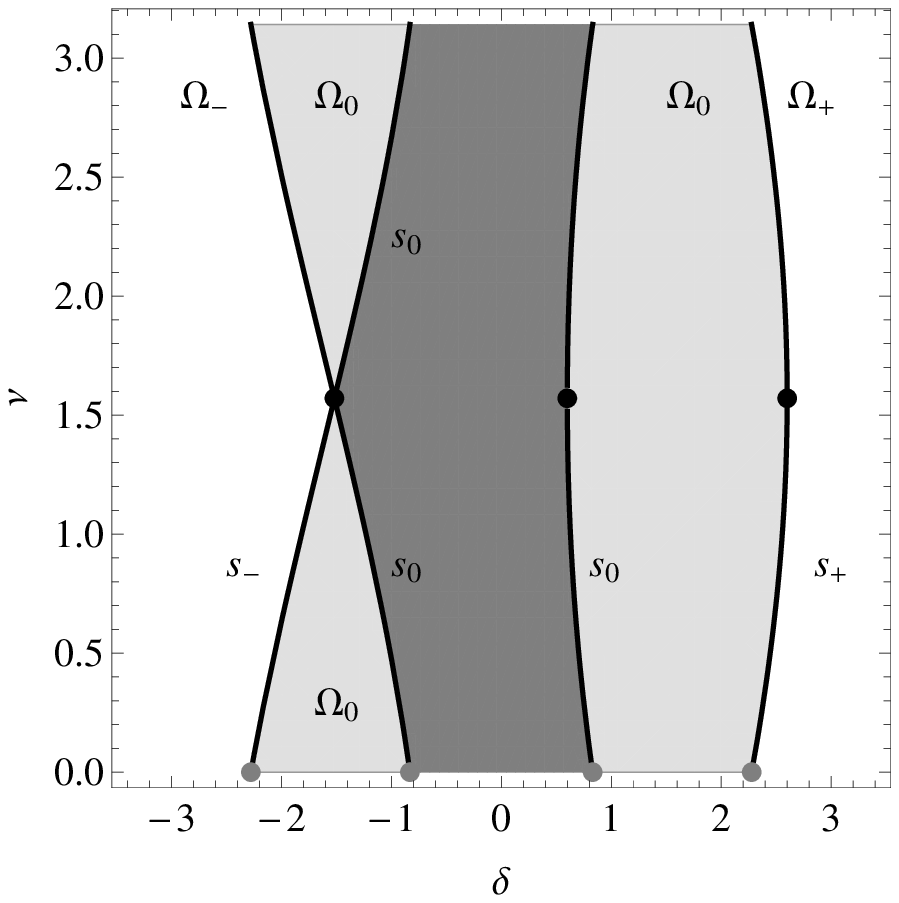}}
\caption{\small Partition of the parameter plane $(\delta,\nu)$.} \label{fig2}
\end{figure}
In this case, the equation $\mathcal P(\sigma;\delta,\nu,\kappa)=0$ has four different roots on the interval $[0,2\pi)$ if $(\delta,\nu)\in\Omega_{-}\cup \Omega_{+}$. If $(\delta,\nu)\in\Omega_{0}$, there are only two different roots (see Fig.~\ref{fig3}, a).
\begin{figure}
\centering
\subfigure[$\kappa=0.4$, $\displaystyle \nu=\frac{\pi}{2}$]{\includegraphics[width=0.3\linewidth]{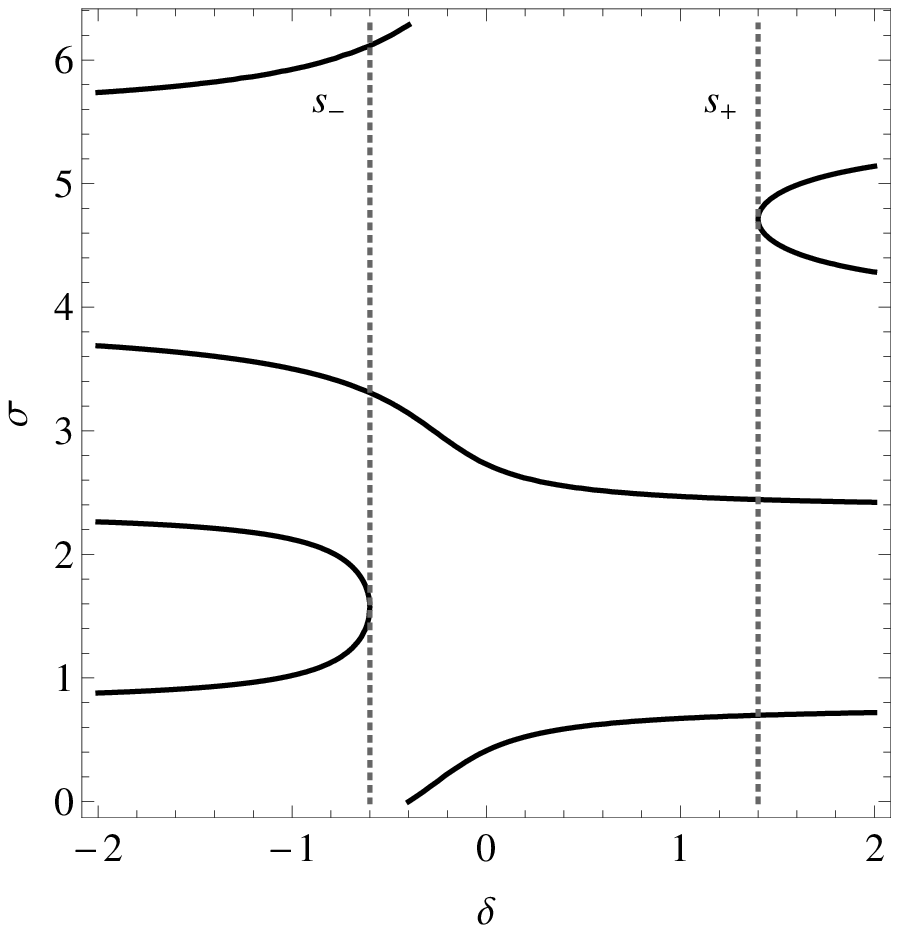}}
\hspace{2ex}
\subfigure[$\kappa=0.9$, $\displaystyle \nu=\frac{\pi}{2}$]{\includegraphics[width=0.3\linewidth]{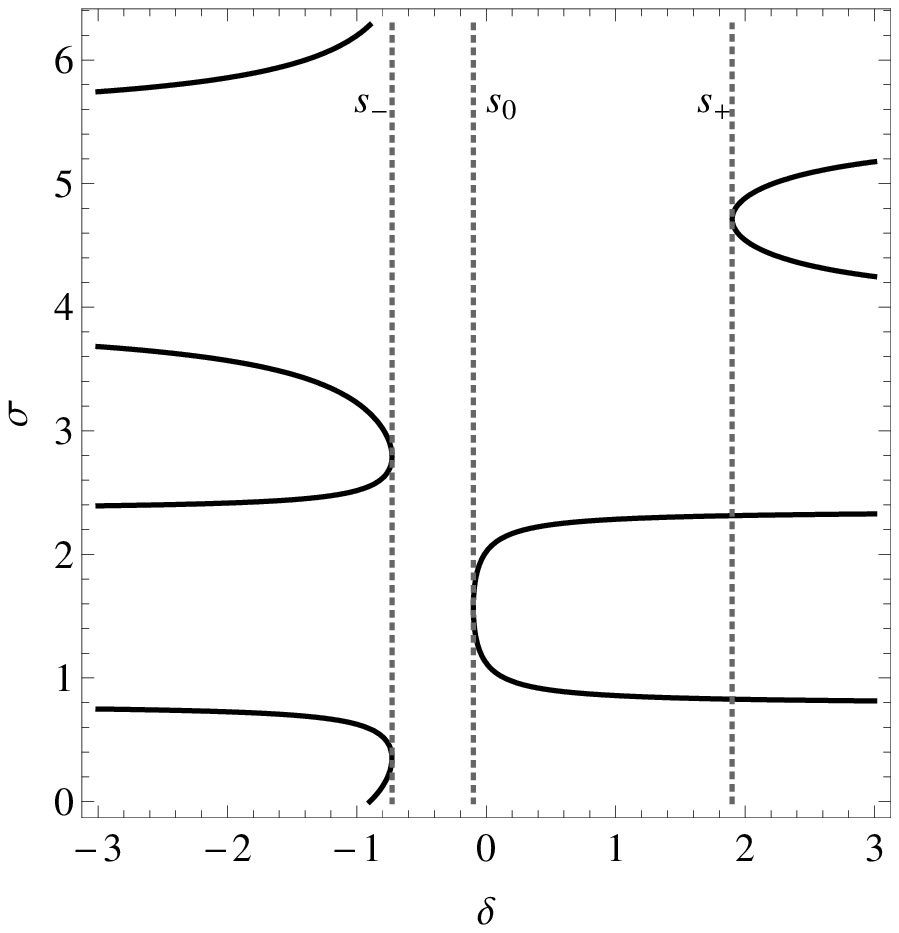}}
\hspace{2ex}
\subfigure[$\kappa=1.6$, $\displaystyle \nu=\frac{\pi}{4}$]{\includegraphics[width=0.3\linewidth]{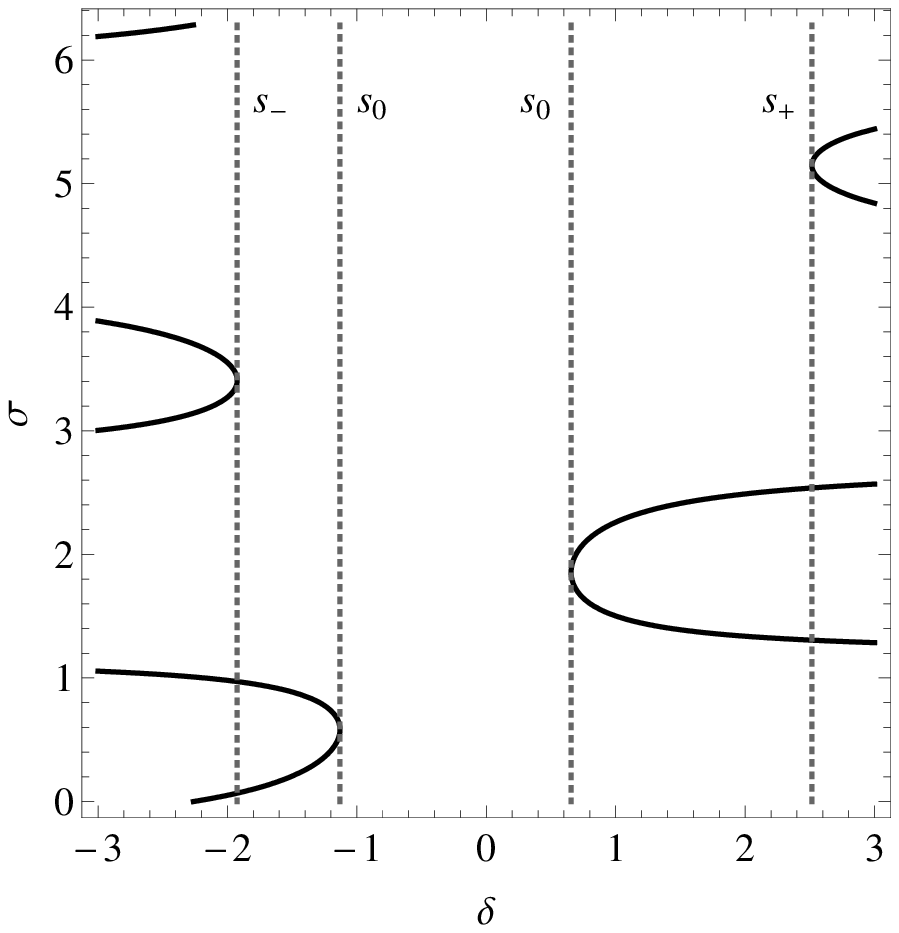}}
\caption{\small The roots to equation \eqref{teq} as functions of the parameter $\delta$. The vertical dotted lines correspond to $s_-$, $s_0$ and $ s_+$.} \label{fig3}
\end{figure}

If $3/4\leq \kappa<1$, there are three curves
\begin{eqnarray*}
    s_+&:=&\{(\delta,\nu)\in [n_2,m_2]\times [0,\pi): \sin\nu=p_1(\delta,\kappa)\},\\
  s_- &:=&\{(\delta,\nu)\in [n_1,m_1]\times [0,\pi): \sin\nu=p_1(\delta,\kappa)\}, \\
  s_0&:=&\{(\delta,\nu)\in [m_1,m_3]\times [0,\pi):\sin\nu=p_1(\delta,\kappa)\},
  \end{eqnarray*}
dividing the parameter plane $(\delta,\nu)$ into four parts (see Fig.~\ref{fig2},b):
\begin{eqnarray*}
    \Omega_{+}&:{=}&\{(\delta,\nu)\in\mathbb R\times[0,\pi): \delta>s_+\},
     \quad
   \Omega_{-}:{=}\{(\delta,\nu)\in\mathbb R\times[0,\pi): \delta<s_-\},\\
   \Omega_\ast&:=&\{\delta\in\mathbb [m_1,m_3], \arcsin p_1(\delta,\kappa)<\nu<\pi-\arcsin p_1(\delta,\kappa)\},\\
   \Omega_0&:=&\mathbb R\times[0,\pi)\setminus\overline{\big (\Omega_{+}\cup \Omega_{-}\cup\Omega_\ast\big)},
\end{eqnarray*}
where $n_1<n_2$ are the roots of the equation $p_1(n,\kappa)=0$, $m_1=-( \sqrt 2 \kappa + \sqrt{2\kappa^2-1})/\sqrt 8$, $m_3=\kappa-1$, $m_2=\kappa+1$
are the roots of the equation $p_1(m,\kappa)=1$.

If $\kappa=1$, the equation $p_1(\delta,\kappa)=0$ has three different roots $n_1<n_3<n_2$. In this case, the parameter plane $(\delta,\nu)$ is divided into the following parts (see Fig.~\ref{fig2},c):
\begin{eqnarray*}
    \Omega_{+}&:{=}&\{(\delta,\nu)\in\mathbb R\times[0,\pi): \delta>s_+\},
     \quad
   \Omega_{-}:{=}\{(\delta,\nu)\in\mathbb R\times[0,\pi): \delta<s_-\},\\
   \Omega_\ast&:=&\{\delta\in\mathbb [m_1,n_3]:\arcsin p_1(\delta,\kappa)<\nu<\pi-\arcsin p_1(\delta,\kappa)\},\\
   \Omega_0&:=&\mathbb R\times[0,\pi)\setminus\overline{\big (\Omega_{+}\cup \Omega_{-}\cup\Omega_\ast\big)},
\end{eqnarray*}
by the curves
\begin{eqnarray*}
    s_+&:=&\{(\delta,\nu)\in [n_2,m_2]\times [0,\pi): \sin\nu=p_1(\delta,\kappa)\},\\
  s_- &:=&\{(\delta,\nu)\in [n_1,m_1]\times [0,\pi): \sin\nu=p_1(\delta,\kappa)\}, \\
  s_0&:=&\{(\delta,\nu)\in [m_1,n_3]\times [0,\pi):\sin\nu=p_1(\delta,\kappa)\}\cup \{\delta=0,\nu\in[0,\pi)\},
  \end{eqnarray*}
where $m_1=-(1+\sqrt 2)/\sqrt8$, $m_2=2$, $p_1(m_{1,2},\kappa)\equiv 1$.

If $\kappa>1$, the equation $p_1(\delta,\kappa)=0$ has four different roots $n_1<n_3<n_4<n_2$ and the parameter plane $(\delta,\nu)$ is divided by the curves
\begin{eqnarray*}
    s_+&:=&\{(\delta,\nu)\in [n_2,m_2]\times [0,\pi): \sin\nu=p_1(\delta,\kappa)\},\\
  s_- &:=&\{(\delta,\nu)\in [n_1,m_1]\times [0,\pi): \sin\nu=p_1(\delta,\kappa)\}, \\
  s_0&:=&\{\delta\in [m_1,n_3]\cup[m_3,n_4], \nu\in[0,\pi): \sin\nu=p_1(\delta,\kappa)\}
  \end{eqnarray*}
 into four parts (see Fig.~\ref{fig2},d):
\begin{eqnarray*}
    \Omega_{+}&:{=}&\{(\delta,\nu)\in\mathbb R\times[0,\pi): \delta>s_+\},
     \quad
   \Omega_{-}:{=}\{(\delta,\nu)\in\mathbb R\times[0,\pi): \delta<s_-\},\\
   \Omega_\ast&:=&\{\delta\in\mathbb [m_1,n_3], \arcsin p_1(\delta,\kappa)<\nu<\pi-\arcsin p_1(\delta,\kappa)\}\cup \{(\delta,\nu)\in\mathbb [n_3,m_3]\times[0,\pi)\}\\
    && \cup\{\delta\in\mathbb [m_3,n_4],  0\leq \nu<\arcsin p_1(\delta,\kappa)\}\cup\{\delta\in\mathbb [m_3,n_4],  \pi-\arcsin p_1(\delta,\kappa)< \nu<\pi\},\\
   \Omega_0&:=&\mathbb R\times[0,\pi)\setminus\overline{\big (\Omega_{+}\cup \Omega_{-}\cup\Omega_\ast\big)},
\end{eqnarray*}
where $
m_1=-( \sqrt 2 \kappa + \sqrt{2\kappa^2-1})/\sqrt8$, $m_2=\kappa+1$, $p_1(m_{1,2},\kappa)\equiv 1$.

As above, if $(\delta,\nu)\in\Omega_{-}\cup \Omega_{+}$, the equation $\mathcal P(\sigma;\delta,\nu,\kappa)=0$ has four different roots on the interval $[0,2\pi)$. If $(\delta,\nu)\in\Omega_{0}$, there are only two different roots. If $(\delta,\nu)\in\Omega_\ast$,  the equation has no solutions (see Fig.~\ref{fig3}, b,c,d).

Thus we have
\begin{Th}\label{Th1}
    If  $(\delta,\nu)\in\Omega_{+}\cup\Omega_{-}$ and $\kappa>0$, system \eqref{MS} has four different solutions with asymptotic expansion in the form of a series \eqref{PAS}.
    If $\kappa>0$ and $(\delta,\nu)\in\Omega_0$, system \eqref{MS} has 2 different solutions with asymptotic expansion in
the form of a series \eqref{PAS}.
\end{Th}
The existence of solutions $\rho_\ast(\tau)$, $\psi_\ast(\tau)$ with the asymptotics \eqref{PAS} as $\tau\geq \tau_\ast$ follows from~\cite{AN89,KF13}. The comparison theorems~\cite{LK14} applied to system \eqref{MS} guarantees that the solutions can be extended to the semi-axis.

\subsection{The roots of multiplicity 2 }
If $(\delta,\nu)\in s_-\cup s_+\cup s_0$, there exists $\sigma$ such that $\mathcal P(\sigma;\delta,\nu,\kappa)=0$ and  $\mathcal P'(\sigma;\delta,\nu,\kappa)=0$. It can easily be checked that $\sigma\in\{\varsigma: \sin \varsigma=z_{1,2}(\delta,\kappa)\}$.
The multiple roots exist if $(\delta,\kappa)\in \mathfrak D_m$, where
$\mathfrak D_m:=\big(\{0\leq p_1(\delta,\kappa)\leq 1\}\cup\{0\leq p_2(\delta,\kappa)\leq 1\}\big)\cap \big(\{|z_1(\delta,\kappa)|\leq 1\}\cup \{ | z_2(\delta,\kappa)|\leq 1\}\big)\cap \{\kappa^2+3\delta^2\geq 3/4\}$ (see Fig.~\ref{fig4}).
\begin{figure}
\centering{
\includegraphics[width=0.4\linewidth]{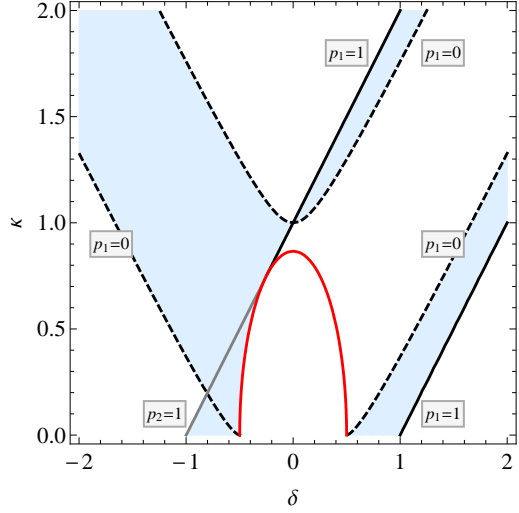}
}
\caption{\small Existence domain (shaded area) of multiple roots to equation \eqref{teq}. } \label{fig4}
\end{figure}
In addition, suppose that $\mathcal P''(\sigma;\delta,\nu,\kappa)\equiv -3\sin\sigma+4\kappa \neq 0$, then $\sigma$ is the root of multiplicity 2.
In this case, $\psi_1$ is determined from the equation:
\begin{gather}
\label{Psi12}    \mathcal P''(\sigma;\delta,\nu,\kappa)\frac{\psi_1^2}{2} = \mathcal C(\sigma), \\
\nonumber    \mathcal C(\sigma)\equiv     \frac{\sin2\sigma-4\lambda^2}{8\lambda\sqrt\lambda}-\Big(\beta_1\sqrt\lambda\sin(2\sigma+\nu)-\alpha_1\sin\sigma +\gamma_1\sqrt\lambda\Big).
\end{gather}
It follows that the asymptotic solution in the form \eqref{PAS} does not exist when
$\mathcal P''(\sigma;\delta,\nu,\kappa) \mathcal C(\sigma)<0$.
 If $\mathcal P''(\sigma;\delta,\nu,\kappa) \mathcal C(\sigma)>0$, equation \eqref{Psi12} has two different roots:
\begin{gather*}
   \psi_1= \pm \phi, \quad  \phi:= \sqrt{ \frac{2 \mathcal C(\sigma)}{\mathcal P''(\sigma;\delta,\nu,\kappa)}}.
\end{gather*}
Note that if $\alpha_1=\beta_1=\gamma_1=0$ and $\lambda>1/2$, then $\mathcal C(\sigma)<0$ for all $\sigma\in\mathbb R$ and a suitable root is $\sigma$ such that $\sin\sigma=z_1(\delta,\kappa)$. The remaining coefficients $\rho_k$, $\psi_k$ are determined from the following recurrent system of equations:
\begin{gather*}
        \begin{split}
            2\sqrt \lambda \rho_k & =\mathcal A_k(\rho_{-1},\dots,\rho_{k-1}, \sigma,\psi_1,\dots,\psi_{k-1}),  \\
            \mathcal P''(\sigma;\delta,\nu,\kappa)\psi_1\psi_k + \frac{\sin\sigma}{\sqrt\lambda} \rho_k& = \mathcal C_k(\rho_{-1},\dots,\rho_{k-1},\sigma, \psi_1,\dots,\psi_{k-1}),
    \end{split}
\end{gather*}
where
\begin{eqnarray*}
    \mathcal C_2 &=&- \frac{\psi_1^3}{6} \mathcal P'''(\sigma;\delta,\nu,\kappa)+ \alpha_1\psi_1\cos\sigma-2\psi_1 (\beta_1\sqrt\lambda+\beta_0\rho_1)\cos(2\sigma+\nu), \\
    \mathcal C_3&=& \frac{\rho_1}{2}-\gamma_1\rho_1-\gamma_2\sqrt\lambda-\frac{\psi_2^2}{2}\mathcal P''(\sigma;\delta,\nu,\kappa) - \frac{\psi_1^2\psi_2}{2}\mathcal P'''(\sigma;\delta,\nu,\kappa)-\frac{\psi_1^4}{24}\mathcal P^{(4)}(\sigma;\delta,\nu,\kappa) \\
    && -\alpha_1\Big(\psi_2 \cos\sigma-\frac{\psi_1^2\sin\sigma}{2}\Big) -\big(2\psi_2(\beta_1\sqrt\lambda+\beta_0\rho_1)+2\psi_1\beta_0\rho_2\big)\cos(2\sigma+\nu)\\
    && +\big(2\psi_1^2 (\beta_1\sqrt\lambda+\beta_0\rho_1)+\alpha_2-\beta_2\sqrt\lambda-\beta_1\rho_1\big)\sin(2\sigma+\nu),
\end{eqnarray*}
etc.

\subsection{The roots of multiplicity 3}

Now suppose  $(\delta,\nu)\in s_-\cup s_+\cup s_0$ and $\sigma$ is the root of equation \eqref{teq} such that  $\mathcal P'(\sigma;\delta,\nu,\kappa)=0$ and $\mathcal P''(\sigma;\delta,\nu,\kappa)=0$. These imply that $\mathcal P'''(\sigma;\delta,\nu,\kappa)\equiv  -3\cos\sigma$, $\sigma\in\{\varsigma: \sin \varsigma=4\kappa/3\}$, $\kappa^2+3\delta^2=3/4$, and $\sin\nu=-\kappa \delta^{-1}(1+32\delta^2)/9$. Let $\mathcal P'''(\sigma;\delta,\nu,\kappa)\neq 0$. In this case, the asymptotic solutions are constructed in the following form:
\begin{gather}
    \label{PAS3}
        \rho (\tau)=\rho_{-1}\sqrt \tau+\rho_0+\sum_{k=1}^{\infty} \rho_k \tau^{-\frac k6}, \quad
        \psi (\tau)=\sigma+\sum_{k=1}^\infty\psi_k\tau^{-\frac k6}, \quad \tau\to\infty.
\end{gather}
It can easily be checked that
\begin{gather*}
    \rho_{-1}=\sqrt\lambda, \quad  \rho_0=\rho_1=\rho_2=0,\quad \rho_3=-\frac{\cos\sigma}{4\lambda},\quad
    \psi_1 =0,\quad    \psi_2= \chi, \\
     \chi :=  \Big[\frac{\mathcal N(\sigma)}{\mathcal P'''(\sigma;\delta,\nu,\kappa) }\Big]^{\frac 13}, \quad
      \mathcal N(\sigma)\equiv \frac{3(\sin 2\sigma-4\lambda^2)}{4\lambda\sqrt\lambda}-6\Big(\beta_1\sqrt\lambda\sin(2\sigma+\nu)-\alpha_1\sin\sigma +\gamma_1\sqrt\lambda\Big).
\end{gather*}
The remaining coefficient $\rho_{k+1}$, $\psi_k$ as $k\geq 3$ are determined from the following chain of equations:
\begin{align*}
            2 \sqrt\lambda \rho_{k+1} & = \mathcal M_{k+1}(\rho_{-1},\dots,\rho_{k}, \sigma,\psi_1,\dots,\psi_{k-1}),  \\
            \mathcal P'''(\sigma;\delta,\nu)\frac{\psi_2^2\psi_k}{2}+\frac{\sin\sigma}{\sqrt\lambda}\rho_{k+1} & = \mathcal N_k(\rho_{-1},\dots,\rho_{k}, \sigma,\psi_1,\dots,\psi_{k-1}),
\end{align*}
where
\begin{align*}
    &\mathcal M_4=0, \quad \mathcal M_5=\frac{\psi_2\sin\sigma}{2\sqrt\lambda}, \quad \mathcal M_6=\frac{\psi_3\sin\sigma}{2\sqrt\lambda}, \quad \mathcal M_7=\frac{\psi_4\sin\sigma-\psi_2^2\cos\sigma}{2\sqrt\lambda}, \\
    & \mathcal N_3=0, \quad \mathcal N_4=-\psi_2(\beta_1\sqrt\lambda+\beta_0\rho_3-\alpha_1)\cos \sigma-\mathcal P^{'''}(\sigma;\delta,\nu,\kappa)\frac{\psi_2 \psi_3^2}{2}-\mathcal P^{(4)}(\sigma;\delta,\nu,\kappa)\frac{\psi_2^4}{24},\\
    & \mathcal N_5=-\psi_3(\beta_1\sqrt\lambda+\beta_0\rho_3-\alpha_1)\cos\sigma-\mathcal P^{'''}(\sigma;\delta,\nu,\kappa)\frac{6\psi_2\psi_3\psi_4+\psi_3^3}{6}-\mathcal P^{(4)}(\sigma;\delta,\nu,\kappa)\frac{\psi_2^3\psi_3}{6},
\end{align*}
etc.
It can easily be checked that this system is solvable whenever $\mathcal N(\sigma)\neq 0$. Note that $\mathcal N(\sigma)\neq 0$ for any $\sigma\in\mathbb R$ if $\alpha_1=\beta_1=\gamma_1=0$ and $\lambda>1/2$.

\subsection{The roots of multiplicity 4}

Let  $(\delta,\nu)\in s_-\cup s_+\cup s_0$ and $\sigma$ be the root of equation \eqref{teq} such that  $\mathcal P'(\sigma;\delta,\nu,\kappa)=0$, $\mathcal P''(\sigma;\delta,\nu,\kappa)=0$ and $\mathcal P'''(\sigma;\delta,\nu,\kappa)=0$. In this case,  $\kappa=3/4$, $\delta=-1/4$, $\nu=\pi/2$, $\sigma=\pi/2$ and $\mathcal P^{(4)}(\sigma;\delta,\nu,\kappa)=3$.
The asymptotic solutions are constructed in the following form:
\begin{gather}
    \label{PAS4}
        \rho (\tau)=\rho_{-1}\sqrt \tau+\rho_0+\sum_{k=1}^{\infty} \rho_k \tau^{-\frac k4}, \quad
        \psi (\tau)=\sigma+\sum_{k=1}^\infty\psi_k\tau^{-\frac k4}, \quad \tau\to\infty.
\end{gather}
It follows easily that
\begin{gather*}
    \rho_{-1}=\sqrt\lambda, \quad  \rho_0=\rho_1=\rho_2=0,\quad
    \psi_1^4 = \mathcal  Q, \quad   \mathcal Q\equiv8 \alpha_1+4\sqrt\lambda(2\beta_1-2\gamma_1-1).
\end{gather*}
If $\mathcal Q> 0$, the last equation has two real, distinct roots: $\psi_1=\pm \xi$, $\xi := \mathcal Q^{1/4}$.
The remaining coefficient $\rho_{k+1}$, $\psi_k$ as $k\geq 2$ are determined from the following chain of equations:
\begin{align*}
             2\sqrt\lambda \rho_{k+1} & = \mathcal T_{k+1}(\rho_{-1},\dots,\rho_{k}, \psi_1,\dots,\psi_{k-1}),  \\
   \psi_1^3\psi_k+\frac{2}{\sqrt\lambda}\rho_{k+1} & = \mathcal Q_k(\rho_{-1},\dots,\rho_{k}, \psi_1,\dots,\psi_{k-1}),
\end{align*}
where
\begin{align*}
    &\mathcal T_3=\frac{\psi_1}{2\sqrt\lambda}, \quad \mathcal T_4=\frac{\psi_2}{2\sqrt\lambda}, \quad \mathcal T_5=\frac{\psi_1^3+3\psi_3}{6\sqrt\lambda}, \quad \mathcal T_6=\frac{\psi_1^2 \psi_2 + \psi_4 - 2 \sqrt \lambda \rho_2^2}{2\sqrt\lambda}, \\
    &\mathcal Q_2=0, \quad \mathcal Q_3=
     - \alpha_1   \psi_1^2 - 4 \beta_1 \sqrt\lambda \psi_1^2 + \frac{\psi_1^6}{24} - \frac{3\psi_1^2 \psi_2^2}{2}  + \frac{\psi_1^2 \rho_2}{\sqrt\lambda},\\
    &\mathcal Q_4=
    -2 \alpha_1 \psi_1 \psi_2 - 8 \beta_1 \sqrt\lambda \psi_1 \psi_2 +\psi_1^5 \psi_2 -\psi_1 \psi_2^3 -3 \psi_1^2 \psi_2 \psi_3 -\frac{2\psi_1 \psi_2 \rho_2 }{\sqrt \lambda} +\frac{ \psi_1^2 \rho_3}{\sqrt\lambda},
\end{align*}
etc.
If $\mathcal Q\leq 0$, the asymptotic solution in the form \eqref{PAS4} does not exists.

Thus we have the following
\begin{Th} Let $(\delta,\nu)\in s_{-} \cup s_{-}\cup s_0$, $\kappa>0$ and $\sigma$ be a root of equation \eqref{teq}.
  \begin{itemize}
  \item If $\mathcal P'(\sigma;\delta,\nu,\kappa)\neq 0$, then  there exists a solution $\rho_\ast(\tau)$, $\psi_\ast(\tau)$ with asymptotic expansion in the form \eqref{PAS} with $\psi_0=\sigma$, $\psi_1=0$, $\psi_2=\theta$.

  \item If $\mathcal P'(\sigma;\delta,\nu,\kappa)=0$ and $\mathcal P''(\sigma;\delta,\nu,\kappa)\mathcal C(\sigma)<0$, then system \eqref{MS} has two solutions $\rho_{\ast}(\tau)$, $\psi_{\ast}(\tau)$ with asymptotic expansion in the form \eqref{PAS} with $\psi_0=\sigma$, $\psi_1=\pm\phi$.

  \item If $\mathcal P'(\sigma;\delta,\nu,\kappa)=0$, $\mathcal P''(\sigma;\delta,\nu,\kappa)=0$ and $\mathcal N(\sigma)\neq 0$, then system \eqref{MS} has solution $\rho_{\ast}(\tau)$, $\psi_{\ast}(\tau)$ with asymptotic expansion in the form \eqref{PAS3} with $\psi_1=0$, $\psi_2=\chi$.

  \item If $\mathcal P'(\sigma;\delta,\nu,\kappa)=0$, $\mathcal P''(\sigma;\delta,\nu,\kappa)=0$,  $\mathcal P'''(\sigma;\delta,\nu,\kappa)=0$  and $\mathcal Q> 0$, then system \eqref{MS} has two solutions $\rho_{\ast}(\tau)$, $\psi_{\ast}(\tau)$ with asymptotic expansion in the form \eqref{PAS4} with $\psi_1=\pm \xi$.
\end{itemize}
\end{Th}
The proof is the same as that of Theorem~\ref{Th1}.

\section{Stability analysis}
\label{sec3}

\subsection{ Linear analysis} Let $\rho_\ast(\tau)$, $\psi_\ast(\tau)$ be one of the particular autoresonant solutions with asymptotics \eqref{PAS}, \eqref{PAS3} or \eqref{PAS4}.
The substitution $\rho(\tau)=\rho_\ast(\tau)+R(\tau)$, $\psi(\tau)=\psi_\ast(\tau)+\Psi(\tau)$ into \eqref{MS} gives the following system with a fixed point at $(0,0)$:
\begin{gather}
\label{ShiftSys}
\begin{split}
&\frac{dR}{d\tau}+\gamma(\tau)R=
    \alpha(\tau)\Big(\sin(\psi_\ast+\Psi)-\sin\psi_\ast\Big) -
    \beta(\tau)\Big( (\rho_\ast+R)\sin(2\psi_\ast+2\Psi+\nu)-\rho_\ast\sin(2\psi_\ast+\nu)\Big), \\
&\frac{d\Psi}{d\tau}=2\rho_\ast R+R^2
+\alpha(\tau)\Big(
\frac{\cos(\psi_\ast+\Psi)}{\rho_\ast+R}-\frac{\cos\psi_\ast}{\rho_\ast}\Big)
- \beta(\tau)\Big(\cos(2\psi_\ast+2\Psi+\nu)-\cos(2\psi_\ast+\nu)\Big).
\end{split}
\end{gather}
Consider the linearized system:
\begin{gather*}
    \frac{d}{d\tau}\begin{pmatrix} R \\ \Psi \end{pmatrix} = {\bf \Lambda}(\tau) \begin{pmatrix}R \\ \Psi \end{pmatrix}, \ \
    {\bf \Lambda}(\tau):=
        \begin{pmatrix}
            \displaystyle  -\gamma(\tau)-\beta(\tau) \sin(2\psi_\ast+\nu) & \displaystyle \alpha(\tau)\cos\psi_\ast-2  \beta(\tau)  \rho_\ast \cos(2\psi_\ast+\nu) \\
            \displaystyle 2\rho_\ast - \frac{\alpha(\tau)\cos\psi_\ast}{\rho_\ast^2}  & \displaystyle -\frac{\alpha(\tau) \sin\psi_\ast}{\rho_\ast} +2 \beta  (\tau)\sin(2\psi_\ast+\nu) \end{pmatrix}.
\end{gather*}
Define
\begin{gather*}
    \hat\rho(\tau):=    \frac{\rho_\ast(\tau)}{\sqrt{\lambda\tau}}-1, \quad   \hat\psi(\tau):=\psi_\ast(\tau)-\sigma,
\end{gather*}
for $\tau\geq 0$,
where $\sigma$ is one of the roots to equation \eqref{teq}. Then the functions $\hat \rho(\tau)$ and $\hat \psi(\tau)$ have the following asymptotics as $\tau\to \infty$: $\hat\rho(\tau)=\mathcal O(\tau^{-1})$ and
\begin{itemize}
  \item $\hat\psi(\tau)=\theta \tau^{-1}+\mathcal O(\tau^{-\frac 32})$ if $\sigma$ is the simple root,
  \item $\hat\psi(\tau)=\pm \phi \tau^{-\frac 12}+\mathcal O(\tau^{-1})$ if $\sigma$ is the root of multiplicity 2,
  \item  $\hat\psi(\tau)=\chi \tau^{-\frac 13}+\mathcal O(\tau^{-\frac 23})$ if $\sigma$ is the root of multiplicity 3,
  \item  $\hat\psi(\tau)=\pm\xi \tau^{-\frac 14}+\mathcal O(\tau^{-\frac 12})$ if $\sigma$ is the root of multiplicity 4.
\end{itemize}
Then the roots of the corresponding characteristic equation $|{\bf \Lambda}(\tau)-z {\bf I}|=0$ can be represented in the form
\begin{gather*}
    z_{\pm}(\tau)=\frac{1}{2}\Big( {\hbox{\rm tr}}  {\bf\Lambda}(\tau)   \pm \sqrt{   D(\tau) }\Big),
\end{gather*}
where
\begin{eqnarray*}
     {\hbox{\rm tr}} {\bf\Lambda}(\tau)  & = &\frac{1}{\sqrt\lambda}\Big(-2\kappa+\mathcal P(\psi_\ast;\delta,\nu,\kappa)\Big)+ \mathcal O(\tau^{-1}), \\
    D(\tau) & :=& \big({\hbox{\rm tr}}{\bf\Lambda}(\tau)\big)^2 -4 {\hbox{\rm det}}  {\bf\Lambda}(\tau)  = -8\tau \sqrt{\lambda}\mathcal P'(\psi_\ast;\delta,\nu,\kappa)+\mathcal O(1)
\end{eqnarray*}
as $\tau\to \infty$.
Therefore, if $\sigma$ is the simple root of \eqref{teq} such that $\mathcal P'(\sigma;\delta,\nu,\kappa)<0$, then $z_\pm(\tau)$ are real of different signs:
\begin{gather*}
z_\pm(\tau)=\pm \tau^{\frac 12} (4\lambda)^{\frac 14} \sqrt{ -\mathcal P'(\sigma;\delta,\nu,\kappa)}+\mathcal O(1),\quad \tau\to\infty.
\end{gather*}
This implies that the fixed point $(0,0)$ of \eqref{ShiftSys} is a saddle in the asymptotic limit, and the corresponding solutions $\rho_\ast(\tau)$, $\psi_\ast(\tau)$ to system \eqref{MS} are unstable (see, for example, \cite{HK02}).

Similarly, if $\sigma$ is the root  of multiplicity 2 such that $\mp\mathcal  P''(\sigma;\delta,\nu,\kappa)>0$ and $\hat\psi(\tau)=\pm\phi\tau^{-1/2}+\mathcal O(\tau^{-1})$, then
\begin{gather*}
z_\pm(\tau)=\pm \tau^{\frac 14}  (4\lambda)^{\frac 14} \sqrt{|\phi\mathcal P''(\sigma;\delta,\nu,\kappa)|}+\mathcal O(1),\quad \tau\to\infty.
\end{gather*}
In this case, the fixed point $(0,0)$ of \eqref{ShiftSys} and the corresponding particular solution to system \eqref{MS} are both unstable.

In the same way, the fixed point of the linearized system is unstable when $\sigma$ is the root  of multiplicity 3 and $\mathcal P'''(\sigma;\delta,\nu)<0$. In this case, the eigenvalues have the following asymptotics:
\begin{gather*}
z_\pm(\tau)=\pm \tau ^{\frac 16} \lambda^{\frac 14}  \sqrt{- \chi^2 \mathcal P'''(\sigma;\delta,\nu,\kappa)}+\mathcal O(1),\quad \tau\to\infty.
\end{gather*}

If $\sigma$ is the root  of multiplicity 4 and $\hat\psi(\tau)=-\xi \tau^{- 1/4}+\mathcal O(\tau^{-1/2})$, then
\begin{gather*}
z_\pm(\tau)=\pm 2 \tau^{\frac 18}  \lambda^{\frac 14} \xi^{\frac 32}+\mathcal O(1),\quad \tau\to\infty,
\end{gather*}
and the corresponding particular solution to system \eqref{MS} is unstable.

Thus we have
\begin{Th} Let $\sigma$ be a root of equation \eqref{teq}.
 \begin{itemize}
   \item If  $\mathcal P'(\sigma;\delta,\nu,\kappa)<0$, the solution $\rho_\ast(\tau)$, $\psi_\ast(\tau)$ with asymptotics \eqref{PAS} is unstable.
   \item If $\mathcal P'(\sigma;\delta,\nu,\kappa)=0$ and $\mp\mathcal  P''(\sigma;\delta,\nu,\kappa)>0$, the solution $\rho_\ast(\tau)$, $\psi_\ast(\tau)$ with asymptotics \eqref{PAS}, $\psi_1=\pm \phi$ is unstable.
   \item If $\mathcal P'(\sigma;\delta,\nu,\kappa)=0$, $\mathcal P''(\sigma;\delta,\nu,\kappa)=0$, and $\mathcal P'''(\sigma;\delta,\nu,\kappa)<0$, the solution $\rho_\ast(\tau)$, $\psi_\ast(\tau)$ with asymptotics \eqref{PAS3} is unstable.
     \item If $\mathcal P'(\sigma;\delta,\nu,\kappa)=0$, $\mathcal P''(\sigma;\delta,\nu,\kappa)=0$, and $\mathcal P'''(\sigma;\delta,\nu,\kappa)=0$, the solution $\rho_\ast(\tau)$, $\psi_\ast(\tau)$ with asymptotics \eqref{PAS4}, $\psi_1=-\xi$ is unstable.
  \end{itemize}
\label{ThExpUnst}
\end{Th}

Let us consider the following cases that are not covered by Theorem~\ref{ThExpUnst}:
\begin{description}
  \item[Case I] $\mathcal P'(\sigma;\delta,\nu,\kappa)>0$.
  \item[Case II] $\mathcal P'(\sigma;\delta,\nu,\kappa)=0$, $\pm \mathcal P''(\sigma;\delta,\nu)>0$ and $\hat\psi(\tau)=\pm \phi\tau^{-\frac 12}+\mathcal O(\tau^{-1})$ as $\tau\to\infty$.
  \item[Case III] $\mathcal P'(\sigma;\delta,\nu,\kappa)=0$, $\mathcal P''(\sigma;\delta,\nu,\kappa)=0$, and $\mathcal P'''(\sigma;\delta,\nu,\kappa)>0$.
  \item[Case IV] $\mathcal P'(\sigma;\delta,\nu,\kappa)=0$, $\mathcal P''(\sigma;\delta,\nu,\kappa)=0$, $\mathcal P'''(\sigma;\delta,\nu,\kappa)=0$  and $\hat\psi(\tau)= \xi\tau^{-\frac 14}+\mathcal O(\tau^{-\frac 12})$ as $\tau\to\infty$.
\end{description}
In these cases, the roots of the characteristic equation are complex. In particular,
\begin{gather*}
    \begin{array}{ll}
      z_\pm(\tau)=\pm i \tau^{\frac 12} (4\lambda)^{\frac 14} \sqrt{\mathcal P'(\sigma;\delta,\nu,\kappa)}+\mathcal O(1),&  \text{ in {\bf Case I}},\\
    z_\pm(\tau)=\pm i  \tau^{\frac 14}(4\lambda)^{\frac14} \sqrt{|\phi \mathcal P''(\sigma;\delta,\nu,\kappa)|}+\mathcal O(1),  &  \text{ in {\bf Case II}},\\
   z_\pm(\tau)=\pm i  \tau^{\frac 16} \lambda^{\frac14}   \sqrt{ \chi^2 \mathcal P'''(\sigma;\delta,\nu,\kappa)}+\mathcal O(1),  &  \text{ in {\bf Case III}},\\
     z_\pm(\tau)= \pm 2 i  \tau^{\frac 18} \lambda^{\frac14}   \xi^{\frac 32}+\mathcal O(1),  &  \text{ in {\bf Case IV}},
\end{array}
\end{gather*}
and $\Re z_\pm(\tau)=\mathcal O(1)$ as $\tau\to\infty$. In such cases, the linear stability analysis fails (see, for example,~\cite{HT94,OS20arxiv}), and the nonlinear terms of the equations must be taken into account.

\subsection{Lyapunov functions}

Let us specify the definition of stability that will be used in this section.
\begin{Def}
The solution $\rho_\ast(\tau)$, $\psi_\ast(\tau)$ to system \eqref{MS} is stable with the weights $\tau^{w_1}$ and $\tau^{w_2}$ as $\tau\geq \tau_0$  if for all $\epsilon>0$ there exists  $\delta_\epsilon>0$ such that for all $(\rho^0, \psi^0)$: $(\rho^0-\rho_\ast(\tau_0))^2 +(\psi^0-\psi_\ast(\tau_0))^2  <\delta_\epsilon^2$ the solution $\rho(\tau)$, $\psi(\tau)$ to system \eqref{MS} with initial data $\rho(\tau_0)=\rho^0$, $\psi(\tau_0)=\psi^0$ satisfies the following inequality:
\begin{gather}
\label{ineq}
(\rho(\tau)-\rho_\ast(\tau))^2 \tau^{2w_1}+(\psi(\tau)-\psi_\ast(\tau))^2 \tau^{2w_2} <\epsilon^2
\end{gather}
for all $\tau>\tau_0$.
\end{Def}

This definition modifies classical concept of stability with $w_1=w_2=0$. Inequality \eqref{ineq} can be considered as the estimate for the norm in the space of continuous functions with the weights $\tau^{w_1}$ and $\tau^{w_2}$. Thus, if the solution is stable with the weights $\tau^{w_1}$, $\tau^{w_2}$ and $w_1,w_2\geq 0$, then the perturbed solutions with initial data sufficiently close to $\rho_\ast(\tau)$, $\psi_\ast(\tau)$ have the following asymptotics: $\rho(\tau)=\rho_\ast(\tau)+\mathcal O(\tau^{-w_1})$, $\psi(\tau)=\psi_\ast(\tau)+\mathcal O(\tau^{-w_2})$ as $\tau\to\infty$.

\begin{Th} Let $\sigma$ be a root of equation \eqref{teq}.
\label{cgs}
 \begin{itemize}
   \item In {\bf Case I}, the solution $\rho_\ast(\tau)$, $\psi_\ast(\tau)$ with asymptotics \eqref{PAS} is stable.
   \item In {\bf Case II}, the solution $\rho_\ast(\tau)$, $\psi_\ast(\tau)$ with asymptotics \eqref{PAS} is stable with the weights $\tau^{3/4}$ and $\tau^{1/2}$.
   \item In {\bf Case III}, the solution $\rho_\ast(\tau)$, $\psi_\ast(\tau)$  with asymptotics \eqref{PAS3} is stable with the weights $\tau^{2/3}$ and $\tau^{1/3}$.
   \item In {\bf Case IV}, the solution $\rho_\ast(\tau)$, $\psi_\ast(\tau)$  with asymptotics \eqref{PAS4} is stable with the weights $\tau^{5/8}$ and $\tau^{1/4}$.
  \end{itemize}
\end{Th}
\begin{proof}
Note that system \eqref{ShiftSys} can be rewritten in a near-Hamiltonian form:
\begin{gather}
    \label{ham}
    \frac{dR}{d\tau}=-\partial_\Psi H(R,\Psi,\tau), \quad
    \frac{d\Psi}{d\tau}=\partial_R H(R,\Psi,\tau)+F(R,\Psi,\tau),
\end{gather}
where
\begin{eqnarray*}
    H(R,\Psi,\tau)  & := &   \rho_\ast  R^2  + \alpha \big(\cos(\psi_\ast+\Psi)-\cos\psi_\ast+\Psi\sin\psi_\ast\big) \\
        &   & - \frac{ \beta \rho_\ast }{2} \Big( \cos(2\psi_\ast+2\Psi+\nu)-\cos(2\psi_\ast+\nu)+2\Psi \sin(2\psi_\ast+\nu)\Big)\\
        &   &   +\gamma R\Psi+ \frac{R^3}{3}-\frac{ \beta R}{2} \Big(\cos(2\psi_\ast+2\Psi+\nu)-\cos(2\psi_\ast+\nu)\Big)
\end{eqnarray*}
and
\begin{eqnarray*}
    F (R,\Psi,\tau) & := & \alpha\left( \frac{\cos(\psi_\ast+\Psi)}{\rho_\ast+R}-\frac{\cos\psi_\ast}{\rho_\ast}\right) -\frac{\beta}{2} \Big(\cos(2\psi_\ast+2\Psi+\nu)-\cos(2\psi_\ast+\nu)\Big)-\gamma \Psi.
\end{eqnarray*}
Taking into account the asymptotic formulas for the particular solutions $\rho_\ast(\tau)=\sqrt{\lambda\tau} (1+\hat \rho(\tau))$ and $\psi_\ast(\tau)=\sigma+\hat\psi(\tau)$, we obtain the following asymptotic estimates:
\begin{eqnarray*}
H & =  &\tau^{\frac 12}\Big( \sqrt\lambda R^2 +  \int\limits_0^\Psi \mathcal P(\sigma+\zeta;\delta,\nu,\kappa)
\,d\zeta \Big)   + \tau^{\frac 12} \hat\psi \Big(\mathcal P(\sigma+\Psi;\delta,\nu,\kappa) - \mathcal P'(\sigma;\delta,\nu,\kappa) \Psi  \Big) \\
&& + \tau^{\frac 12} \frac{\hat\psi^2}{2} \Big(\mathcal P'(\sigma+\Psi;\delta,\nu,\kappa) - \mathcal P''(\sigma;\delta,\nu,\kappa) \Psi  -\mathcal P'(\sigma;\delta,\nu,\kappa)  \Big)  \\
&&+ \tau^{\frac 12} \frac{\hat\psi^3}{6} \Big(\mathcal P''(\sigma+\Psi;\delta,\nu,\kappa) - \mathcal P'''(\sigma;\delta,\nu,\kappa) \Psi  -\mathcal P''(\sigma;\delta,\nu,\kappa)  \Big) \\
&& +\gamma_0 R\Psi+\frac{R^3}{3}-\frac{\beta_0 R}{2}\Big(\cos(2\sigma+2\Psi+\nu)-\cos(2\sigma+\nu)\Big)+\\
&&+ \beta_0 R \hat\psi \Big(\sin(2\sigma+2\Psi+\nu)-\sin(2\sigma+\nu)\Big)+\mathcal O(\hat\psi^2)+\mathcal O(\tau^{\frac 12}\hat\psi^4 )+\mathcal O(\tau^{-\frac 12})
\end{eqnarray*}
and
\begin{eqnarray*}
F & =  & \frac{1}{\sqrt\lambda} \int\limits_0^\Psi \Big(\mathcal P(\sigma+\zeta;\delta,\nu,\kappa)-2\kappa\Big)\,d\zeta +\mathcal O(\hat\psi)+\mathcal O(\tau^{-\frac 12})
\end{eqnarray*}
as $\tau\to\infty$, for all  $ (R,\Psi)\in B(d_\ast)$, where
\begin{gather*}
    B(d_\ast):=\{(R,\Psi)\in\mathbb R^2:  d=\sqrt{R^2+\Psi^2}\leq d_\ast\}, \quad d_\ast={\hbox{\rm const}}>0.
\end{gather*}

Consider {\bf Case I}, when $\sigma$ is the simple root to equation \eqref{teq}. In this case, $\hat\psi=\theta \tau^{-1}+\mathcal O(\tau^{-3/2})$,
and we get
\begin{gather*}
   H= \tau^{\frac 12}\Big(\sqrt\lambda R^2 +  \omega_1^2 \frac{\Psi^2}{2}+\mathcal O(d^3)\Big)+\mathcal O(\tau^{-\frac 12}d^2),\quad F= -2\gamma_0\Psi +\mathcal O(d^2)+\mathcal O(\tau^{-\frac 12}d)
\end{gather*}
as $d\to0$ and $\tau\to\infty$, where $\omega_1^2=\mathcal P'(\sigma;\delta,\nu,\kappa)>0$. Note that the asymptotic estimates are uniform with respect to $(R,\Psi,\tau)$ in the domain
$\{(R,\Psi,\tau)\in\mathbb R^3: d\leq d_\ast, \tau\geq \tau_\ast\}$,
where $d_\ast$, $\tau_\ast={\hbox{\rm const}}>0$.
Consider the combination
\begin{gather*}
    V_1(R,\Psi,\tau):= \tau^{-\frac 12}\Big( H(R,\Psi,\tau)-\gamma_0 R\Psi\Big)
\end{gather*}
as a Lyapunov function candidate for system \eqref{ham}.
The function $V_1(R,\Psi,\tau)$ has the following asymptotics:
\begin{gather*}
V_1 =  \sqrt\lambda R^2 +  \omega_1^2 \frac{\Psi^2}{2}+\mathcal O(d^3)+\mathcal O(\tau^{-\frac 12}d^2)
\end{gather*}
as $d\to 0$ and $\tau\to\infty$.
It can easily be checked that for all $0<\varkappa <1$ there exist $d_1>0$ and $\tau_{11}>0$ such that
\begin{gather*}
 (1-\varkappa ) W_1(R,\Psi)  \leq V_1(R,\Psi,\tau) \leq  (1+\varkappa ) W_1(R,\Psi)
\end{gather*}
for all $(R,\Psi,\tau)\in D_{W_1}(d_1,\tau_{11})$, where
\begin{gather*}
D_{W_1}(d_1,\tau_{11}) :=  \{(R,\Psi,\tau)\in\mathbb R^3: W_1(R,\Psi)\leq d_1^2, \tau\geq \tau_{11}\}, \quad W_1(R,\Psi) := \sqrt \lambda R^2+\omega_1^2\frac{\Psi^2}{2}.
\end{gather*}
The total derivative of the function $V_1(R,\Psi,\tau)$ with respect to $\tau$ along the trajectories of system \eqref{ham} has the form:
\begin{gather*}
\frac{dV_1}{d\tau}\Big|_\eqref{ham} =\frac{\partial V_1}{\partial \tau}-\frac{\partial V_1}{\partial R}\partial_\Psi H+ \frac{\partial V_1}{\partial \Psi} \Big(\partial_R H+F\Big)= -  2 \gamma_0 \Big( \sqrt \lambda R^2 +\omega_1^2 \frac{\Psi^2}{2} +\mathcal O(d^3)\Big)+\mathcal O(\tau^{-\frac 12}d^2)
\end{gather*}
as $d\to 0$ and $ \tau\to\infty$.
Hence, for all $0<\varkappa <1$ there exist $d_2>0$ and $\tau_{12}>0$ such that
\begin{gather}
\label{V1est}
    \frac{dV_1}{d\tau}\Big|_\eqref{ham} \leq -  2\gamma_\varkappa  V_1, \quad \gamma_\varkappa :=\gamma_0\Big(\frac{1-\varkappa }{1+\varkappa }\Big)>0,
\end{gather}
 for all $(R,\Psi,\tau)\in  D_{W_1}(d_2,\tau_{12})$. Thus, for all $0<\varepsilon< d_0$ there exist $\delta_\varepsilon := \varepsilon \sqrt{(1-\varkappa)/(1+\varkappa)}/2$  such that
\begin{gather*}
\sup_{W_1\leq \delta_\varepsilon^2} V_1(R,\Psi,\tau)\leq (1+\varkappa )  \delta^2_\varepsilon< (1-\varkappa )   \varepsilon^2 \leq \inf_{W_1=\varepsilon^2}V_1(R,\Psi,\tau)
\end{gather*}
for all $\tau>\tau_0$, where $d_0=\min\{d_1,d_2\}$ and $\tau_0=\max\{\tau_{11},\tau_{12}\}$. The last estimates and the negativity of the total derivative of the function $V_1(R,\Psi,\tau)$ ensure that any solution of system \eqref{ham} with initial data $ W_1(R(\tau_0), \Psi(\tau_0))\leq \delta_\varepsilon^2 $ cannot leave the domain $\{(R,\Psi)\in\mathbb R^2: W_1(R,\Psi)\leq \varepsilon^2\}$ as  $\tau>\tau_0$. Hence, the fixed point $(0,0)$ is stable as $\tau>\tau_0$. The stability on the finite time interval $(0,\tau_0]$ follows from the theorem on the continuity of the solutions to the Cauchy problem with respect to the initial data.

Consider {\bf Case II}.
Let $\sigma$ be a root of multiplicity 2 to equation \eqref{teq} such that $\mathcal P'(\sigma;\delta,\nu,\kappa)=0$, $\pm\mathcal P''(\sigma;\delta,\nu,\kappa)\hat\psi=\omega_2^2 \tau^{-1/2}+\mathcal O(\tau^{-1})$, where $\omega_2^2=\phi \mathcal P''(\sigma;\delta,\nu,\kappa)>0$.
Hence,
\begin{gather*}
H=\tau^{\frac 12} \Big(\sqrt \lambda  R^2+\mathcal P''(\sigma;\delta,\nu,\kappa)\frac{\Psi^3}{6}+\mathcal O(d^4)\Big) +\omega_2^2\frac{\Psi^2}{2}  + \mathcal O(d^3)+\mathcal O(\tau^{-\frac 12}d^2)
\end{gather*}
as $d\to 0$ and $\tau\to\infty$. In this case, $H(R,\Psi,\tau)$ is sign indefinite and  the function $V_1(R,\Psi,\tau)$  can not be used as a Lyapunov function.
Consider the change of variables
\begin{gather}
\label{chvar2}
 R(\tau)=\tau^{-\frac 34} r(\tau), \quad \Psi(\tau)= \tau^{-\frac 12}\varphi(\tau)
\end{gather}
in system \eqref{ham}. The transformed system is
\begin{gather}
    \label{ham1}
    \frac{d r}{d\tau}=-\partial_{\varphi} H_{2}(r,\varphi,\tau), \quad
    \frac{d \varphi}{d\tau}=\partial_{r} H_{2}(r,\varphi,\tau)+  F_{2}(r,\varphi,\tau),
\end{gather}
where
\begin{gather*}
    H_{2}(r,\varphi,\tau):=\tau^{\frac 54}  H ( \tau^{-\frac 34}r,  \tau^{-\frac 12} \varphi, \tau) -\tau^{-1}\frac{3 r\varphi}{4}, \quad
    F_{2}(r,\varphi,\tau):=\tau^{\frac 12}F (\tau^{-\frac 34}r, \tau^{-\frac 12}\varphi,\tau) + \tau^{-1}\frac{5\varphi}{4}.
 \end{gather*}
Taking into account \eqref{PAS}, we see that
\begin{gather*}
H_2 =\tau^{\frac 14}\Big(\sqrt \lambda R^2+\omega_2^2 \frac{\varphi^2}{2}+\mathcal P''(\sigma;\delta,\nu,\kappa)\frac{\varphi^3}{6}\Big)+\mathcal O(\Delta^2), \quad
F_2=-2\gamma_0 \varphi+\mathcal O(\Delta \tau^{-\frac 12})
\end{gather*}
as $t \to \infty$ and $ \Delta=\sqrt{r^2 +  \varphi^2}\to 0$. Note that the function $H_2(r,\varphi,\tau)$ is suitable for the basis of a Lyapunov function candidate:
\begin{gather*}
V_2(r,\varphi,\tau) = \tau^{-\frac 14}\Big(H_2(r,\varphi,\tau)-\gamma_0 r\varphi\Big)
\end{gather*}
It follows easily that for all $0<\varkappa <1$ there exist $ \Delta_1>0$ and $ \tau_1>0$ such that
\begin{gather*}
   (1-\varkappa )  W_2(r,\varphi) \leq V_2(r,\varphi,\tau) \leq   (1+\varkappa )  W_2(r,\varphi)
\end{gather*}
for all $(r,\varphi,\tau)\in  D_{W_2}(\Delta_1,\tau_1)$, where $
 W_2(r,\varphi):= \sqrt\lambda r^2+\omega_2^2{\varphi^2}/{2}$.
The total derivative of the function $V_2(r,\varphi,\tau)$ has a sign definite leading term of the asymptotics:
\begin{align*}
\begin{split}
    \frac{d V_2}{d\tau}\Big|_{\eqref{ham1}} & =  \frac{\partial  V_2}{\partial \tau} +  \frac{\partial V_{2}}{\partial \varphi} F_{2}+\gamma_0\tau^{-\frac 14}\Big(\varphi \frac{\partial  H_{2}}{\partial \varphi} - r\frac{\partial H_{2}}{\partial r}\Big)\\
& =    -2\gamma_0\Big(\sqrt \lambda r^2+\omega_2^2\frac{\varphi^2}{2}+ \mathcal O(\Delta^3)\Big)+\mathcal O(\tau^{-\frac 14} \Delta^2)
\end{split}
\end{align*}
as $\Delta\to 0$ and $\tau\to\infty$. Hence, for all $0<\varkappa <1$ there exist $\Delta_2>0$ and $\tau_2>0$ such that
\begin{gather}
    \label{V2est}
    \frac{dV_2}{d\tau}\Big|_\eqref{ham1} \leq -2\gamma_\varkappa  V_2 \leq 0
\end{gather}
for all $(r,\varphi,\tau)\in  D_{W_2}(\Delta_2,\tau_2)$. As above, the last inequality implies the stability of the equilibrium $(0,0)$ to system \eqref{ham1}. Returning to the original variables, we obtain the result of the Theorem.

In {\bf Case III}, $\sigma$ is a root of multiplicity 3 to equation \eqref{teq} such that $\mathcal P'(\sigma;\delta,\nu,\kappa)=\mathcal P''(\sigma;\delta,\nu,\kappa)=0$ and $\mathcal P'''(\sigma;\delta,\nu,\kappa)>0$.
The solution $\rho_\ast(\tau)$, $\psi_\ast(\tau)$ has the asymptotics \eqref{PAS3}, and $\hat\psi=\chi \tau^{-1/3}+\mathcal O(\tau^{-2/3})$.
In this case,
\begin{eqnarray*}
H(R,\Psi,\tau)&=&\tau^{\frac 12} \Big(\sqrt \lambda R^2+\mathcal P'''(\sigma;\delta,\nu,\kappa)\frac{\Psi^4}{24} +\mathcal O(d^5)\Big)+ \mathcal O(\tau^{\frac 16}  d^3) +\mathcal O(d^2)
\end{eqnarray*}
as $d\to 0$ and $\tau\to\infty$. Note that this function is sign indefinite in a neighborhood of the equilibrium and can not be used in the construction of a Lyapunov function. Indeed, if $\Psi\sim \epsilon^{1/4}$ and $\tau\sim \epsilon^{-1}$ as $\epsilon\to 0$, the leading and the remainder terms in the last expression can be of the same order. Consider the change of variables:
\begin{gather*}
 R(\tau)=\tau^{-\frac 23} r(\tau), \quad \Psi(\tau)= \tau^{-\frac 13} \varphi(\tau)
\end{gather*}
in system \eqref{ham}. It can easily be checked that the transformed system has the form
\begin{gather}
    \label{ham2}
    \frac{d r}{d\tau}=-\partial_{\varphi} H_3(r,\varphi,\tau), \quad
    \frac{d \varphi}{d\tau}=\partial_{r}  H_3(r,\varphi,\tau)+ F_3(r,\varphi,\tau),
\end{gather}
where
\begin{gather*}
H_3(r,\varphi,\tau):=\tau H ( \tau^{-\frac 23}r,  \tau^{-\frac 13}\varphi, \tau )-\tau^{-1}\frac{2r\varphi}{3}, \quad
 F_3(r,\varphi,\tau):=\tau^{\frac 13}  F ( \tau^{-\frac 23}r,  \tau^{-\frac 13}\varphi, \tau )+\tau^{-1}\varphi.
 \end{gather*}
Using asymptotic formulas for the particular solution, we obtain
 \begin{gather*}
 H_3 =   \tau^{\frac16}\Big(\sqrt \lambda r^2 +  \omega_3^2 \frac{\varphi^2}{2}-\mathcal P'''(\sigma;\delta,\nu,\kappa)\frac{\varphi^3}{6}\big(\chi-\frac{\varphi}{4}\big)\Big)+\mathcal O(\Delta^2),  \quad
F_3  =  -2\gamma_0 \varphi+\mathcal O(\Delta \tau^{-\frac 13})
\end{gather*}
as $\Delta\to 0$ and $\tau\to\infty$, where $\omega_3^2:=\chi^2 \mathcal P'''(\sigma;\delta,\nu)/2>0$.
Consider the combination
\begin{gather*}
V_3(r,\varphi,\tau) = \tau^{-\frac 16}\Big(H_3(r,\varphi,\tau)-\gamma_0 r \varphi\Big)
\end{gather*}
as a Lyapunov function candidate for system \eqref{ham2}. It follows easily that for all $0<\varkappa <1$ there exist $ \Delta_3>0$ and $ \tau_3>0$ such that
\begin{gather*}
   (1-\varkappa )  W_3(r,\varphi) \leq V_3(r,\varphi,\tau) \leq   (1+\varkappa )  W_3(r,\varphi)
\end{gather*}
for all $(r,\varphi,\tau)\in  D_{W_3}({\Delta_3,\tau_3})$, where
$W_3(r,\varphi) := \sqrt{\lambda}r^2+\omega_3^2 \varphi^2/2$.
The derivative of this function with respect to $\tau$ along the trajectories of system \eqref{ham2} satisfies:
\begin{eqnarray*}
\frac{d V_3}{d\tau}\Big|_{\eqref{ham2}} & = &   -2\gamma_0 \Big(\sqrt \lambda r^2+\omega_3^2\frac{\varphi^2}{2}+ \mathcal O(\Delta^3)\Big)+\mathcal O(\tau^{-\frac 16} \Delta^2)
\end{eqnarray*}
as $\Delta\to 0$ and $\tau\to \infty$. Hence, for all $0<\varkappa <1$ there exist $\Delta_4>0$ and $\tau_4>0$ such that
\begin{gather}
\label{V3est}
    \frac{d V_3}{d\tau}\Big|_\eqref{ham2} \leq - 2 \gamma_\varkappa  V_3\leq 0
\end{gather}
for all $(r,\varphi,\tau)\in  D_{W_3}(\Delta_4,\tau_4)$. The last inequality implies that the fixed point $(0,0)$ of system \eqref{ham2} is stable and the solution $\rho_\ast(\tau)$, $\psi_\ast(\tau)$ is stable with the weights $\tau^{2/3}$ and $\tau^{1/3}$.

Finally, consider {\bf Case IV}. The change of variables
\begin{gather*}
%\label{chvar3}
 R(\tau)=\tau^{-\frac 58} r(\tau), \quad \Psi(\tau)= \tau^{-\frac 14} \varphi(\tau)
\end{gather*}
transforms system \eqref{ham} into
\begin{gather}
    \label{ham3}
    \frac{d r}{d\tau}=-\partial_{\varphi} H_4(r,\varphi,\tau), \quad
    \frac{d\varphi}{d\tau}=\partial_{r} H_4(r,\varphi,\tau)+F_4(r,\varphi,\tau),
\end{gather}
where
\begin{gather*}
     H_{4}(r,\varphi,\tau):=\tau^{\frac 78} H(\tau^{-\frac 58}r,\tau^{-\frac 14}\varphi,\tau) -\tau^{-1}\frac{5r\varphi}{8}, \quad
     F_{4}(r,\varphi,\tau):=\tau^{-1}\frac{7 \varphi }{8} +\tau^{\frac 14}
     F(\tau^{-\frac 58}r,\tau^{-\frac 14}\varphi,\tau).
 \end{gather*}
Taking into account \eqref{PAS4}, it can easily be checked that
\begin{gather*}
    H_{4}(r,\varphi,\tau) = \tau^{\frac 18}\Big(\sqrt\lambda R^2+\omega_4^2 \frac{\varphi^2}{2}+\frac{\xi^2\varphi^3}{12}+\frac{\xi\varphi^4}{8}+\frac{\varphi^5}{40}\Big)+\mathcal O(\Delta^2), \quad
    F_{4}(r,\varphi,\tau) = -2\gamma_0 \varphi+\mathcal O(\Delta \tau^{-\frac 18})
\end{gather*}
as $\Delta\to 0$ and $\tau\to \infty$, where $\omega_4^2:=\xi^3/4>0$. Consider
\begin{gather*}
V_4(r,\varphi,\tau)=\tau^{-\frac 18} \Big(H_4(r,\varphi,\tau)-\gamma_0 r\varphi\Big)
\end{gather*}
as a Lyapunov function candidate to system \eqref{ham3}.
We see that  for all $0<\varkappa <1$ there exist $ \Delta_5>0$ and $ \tau_5>0$ such that
\begin{gather*}
   (1-\varkappa )  W_4(r,\varphi) \leq V_4(r,\varphi,\tau) \leq   (1+\varkappa )  W_4(r,\varphi)
\end{gather*}
for all $(r,\varphi,\tau)\in  D_{W_4}({\Delta_5,\tau_5})$, where
$W_4(r,\varphi) := \sqrt{\lambda}r^2+\omega_4^2 \varphi^2/2$.
The total derivative of $V_4$ with respect to $\tau$ satisfies:
\begin{eqnarray*}
\frac{d V_4}{d\tau}\Big|_{\eqref{ham3}} & = &   -2\gamma_0 \Big(\sqrt \lambda r^2+\omega_4^2\frac{\varphi^2}{2}+ \mathcal O(\Delta^3)\Big)+\mathcal O(\tau^{-\frac 18} \Delta^2)
\end{eqnarray*}
as $\Delta\to 0$ and $\tau\to \infty$. Hence, for all $0<\varkappa <1$ there exist $\Delta_6>0$ and $\tau_6>0$ such that
\begin{gather*}
    \frac{d V_4}{d\tau}\Big|_\eqref{ham3} \leq - 2 \gamma_\varkappa  V_4\leq 0
\end{gather*}
for all $(r,\varphi,\tau)\in  D_{W_4}(\Delta_6,\tau_6)$. The last inequality implies that the fixed point $(0,0)$ of system \eqref{ham2} is stable. Returning to the variables $(\rho,\psi)$ we obtain the result of the theorem.

\end{proof}

Thus, the particular autoresonant solutions $\rho_\ast(\tau)$, $\psi_\ast(\tau)$ with power-law asymptotics are stable in all four cases. The stability ensures the existence of a family of solutions with a similar behaviour. In particular, we have the following.

\begin{Cor}
There exist $\Delta_\ast>0$ and $T_\ast>0$ such that for all $(\rho^0,\psi^0)$: $(\rho^0-\rho_\ast(T_\ast))^2+(\psi^0-\psi_\ast(T_\ast))^2<\Delta_\ast^2$ the solution $\rho(\tau)$, $\psi(\tau)$ to system \eqref{MS} with initial data $\rho(T_\ast)=\rho^0$, $\psi(T_\ast)=\psi^0$ has the following estimates as $\tau\to\infty${\rm :}
\begin{align}
\label{as1}& \rho=\sqrt{\lambda\tau}+\mathcal O(\tau^{-\frac 12}), \ \  \psi=\sigma+\mathcal O(\tau^{-1}) \quad \text{in {\bf Case I}};  \\
\label{as2} & \rho=\sqrt{\lambda\tau}+\mathcal O(\tau^{-\frac 12}),  \ \  \psi=\sigma+\mathcal O(\tau^{-\frac 12}) \quad \text{in {\bf Case II}}; \\
\label{as3} & \rho=\sqrt{\lambda\tau}+\mathcal O(\tau^{-\frac 12}),  \ \  \psi=\sigma+\mathcal O(\tau^{-\frac 13}) \quad \text{in {\bf Case III}};
\\
\label{as4} & \rho=\sqrt{\lambda\tau}+\mathcal O(\tau^{-\frac 12}),  \ \  \psi=\frac{\pi}{2}+\mathcal O(\tau^{-\frac 14}) \quad \text{in {\bf Case IV}}.
\end{align}
\end{Cor}
\begin{proof}
Consider {\bf Case I}. Let $R(\tau)$, $\Psi(\tau)$  be a solution to system \eqref{ham} starting from the ball $B(\Delta_\ast)$ at $\tau=T_\ast$, where $\Delta_\ast=d_0$ and $T_\ast=\tau_0$ (see Theorem~\ref{cgs}). Then it follows from \eqref{V1est} that the function $v(\tau)=V_1(R(\tau),\Psi(\tau),\tau)$ satisfies the inequality:
\begin{gather}\label{vest}
        \frac{d v}{d\tau}\leq -2 \gamma_\varkappa v
\end{gather}
as $\tau\geq T_\ast$, where $\gamma_\varkappa=\gamma_0(1-\varkappa)/(1+\varkappa)$, $\varkappa\in (0,1)$.
Integrating the last expression with respect to $\tau$, we obtain $0\leq v(\tau)\leq v(T_\ast) \exp\big(-2\gamma_\varkappa(\tau-T_\ast)\big)$, where $0\leq v(T_\ast) \leq C_\ast\Delta_\ast^2$, $C_\ast={\hbox{\rm const}}$. Thus we have $R(\tau)=\mathcal O(e^{-\gamma_\varkappa \tau})$, $\Psi(\tau)=\mathcal O(e^{-\gamma_\varkappa \tau})$ as $\tau\to\infty$.  Since $\rho(\tau)=\rho_\ast(\tau)+R(\tau)$, $\psi(\tau)=\psi_\ast(\tau)+\Psi(\tau)$ and $\rho_\ast(\tau)=\sqrt{\lambda\tau}+\mathcal O(\tau^{-1/2})$, $\psi_\ast(\tau)=\sigma+\mathcal O(\tau^{-1})$ as $\tau\to\infty$, we have \eqref{as1}.

{\bf Case II}. Let $r(\tau)$, $\varphi(\tau)$ be a solution to system \eqref{ham1} starting from $\{(r_0,\varphi_0): W_2(r_0,\varphi_0)\leq \Delta_\ast^2\}$ at $\tau=T_\ast$, where $\Delta_\ast=\min\{\Delta_1,\Delta_2\}$, $T_\ast=\max\{\tau_1,\tau_2\}$. From \eqref{V2est} it follows that the derivative of the function $v(\tau)=V_2\big(r (\tau),\varphi (\tau),\tau\big)$ satisfies \eqref{vest} as $\tau\geq T_\ast$.
By integrating this estimate with respect to $\tau$ and taking into account \eqref{chvar2}, we get $R(\tau)=\mathcal O(\tau^{-3/4}e^{-\gamma_\varkappa \tau})$, $\Psi (\tau)=\mathcal O(\tau^{-1/2}e^{-\gamma_\varkappa \tau})$ as $\tau\to\infty$.  Returning to the variables $(\rho,\psi)$ and using \eqref{PAS}, we derive \eqref{as2}.

{\bf Case III}. Let $r (\tau)$, $\varphi (\tau)$ be a solution to system \eqref{ham2} with initial data from  $\{(r_0,\varphi_0): W_3(r_0,\varphi_0)\leq \Delta_\ast^2\}$ at $\tau=T_\ast$, where $\Delta_\ast=\min\{\Delta_3,\Delta_4\}$, $T_\ast=\max\{\tau_3,\tau_4\}$. From \eqref{V3est} it follows that the function $v(\tau)=V_3\big(r (\tau),\varphi (\tau),\tau\big)$ satisfies  \eqref{vest} as $\tau\geq T_\ast$. Hence, $R(\tau)=\mathcal O(\tau^{-2/3}e^{-\gamma_\varkappa \tau})$, $\Psi (\tau)=\mathcal O(\tau^{-1/3}e^{-\gamma_\varkappa \tau})$ as $\tau\to\infty$.  Combining this with \eqref{PAS3}, we get \eqref{as3}.

Finally, consider {\bf Case IV}. Let $r (\tau)$, $\varphi (\tau)$ be a solution to system \eqref{ham3} with initial data from  $\{(r_0,\varphi_0): W_4(r_0,\varphi_0)\leq \Delta_\ast^2\}$ at $\tau=T_\ast$, where $\Delta_\ast=\min\{\Delta_5,\Delta_6\}$, $T_\ast=\max\{\tau_5,\tau_6\}$. It follows easily that the function $v(\tau)=V_4\big(r (\tau),\varphi (\tau),\tau\big)$ satisfies  \eqref{vest} as $\tau\geq T_\ast$. Hence, $R(\tau)=\mathcal O(\tau^{-5/8}e^{-\gamma_\varkappa \tau})$, $\Psi (\tau)=\mathcal O(\tau^{-1/4}e^{-\gamma_\varkappa \tau})$ as $\tau\to\infty$.  Combining this with \eqref{PAS4}, we get \eqref{as4}.
\end{proof}

\section{Conclusion}

Thus, we have described possible autoresonant modes in oscillating systems with a combined excitation and a weak dissipation. The presence of dissipation in the model leads to exponential stability of a part of autoresonant modes in comparison with a systems without dissipation, where only polynomial stability takes place~\cite{OS18}.
Depending on the values of the excitation parameters and the dissipation coefficient, the system can have different number of autoresonant modes with different phase mismatch:  $\psi(\tau)\sim \sigma$ as $\tau\to\infty$, where $\sigma$ is the root of equation \eqref{teq}. For every $\kappa=\gamma_0\sqrt\lambda >0$ the curves $s_\pm$ and $s_0$ consist of bifurcations points on the plane $(\delta,\nu)$, where $\delta=\beta_0 \sqrt\lambda$.  Outside of this curves the roots of equation \eqref{teq} are simple and system \eqref{MS} can have two or four different autoresonant modes.
Their stability depends on the sign of the value $\mathcal P'(\sigma;\delta,\nu,\kappa)$ (see the shaded areas in Fig.~\ref{fig33}).
\begin{figure}
\centering
\subfigure[$\kappa=0.4$, $\displaystyle \nu=\frac{\pi}{2}$]{\includegraphics[width=0.3\linewidth]{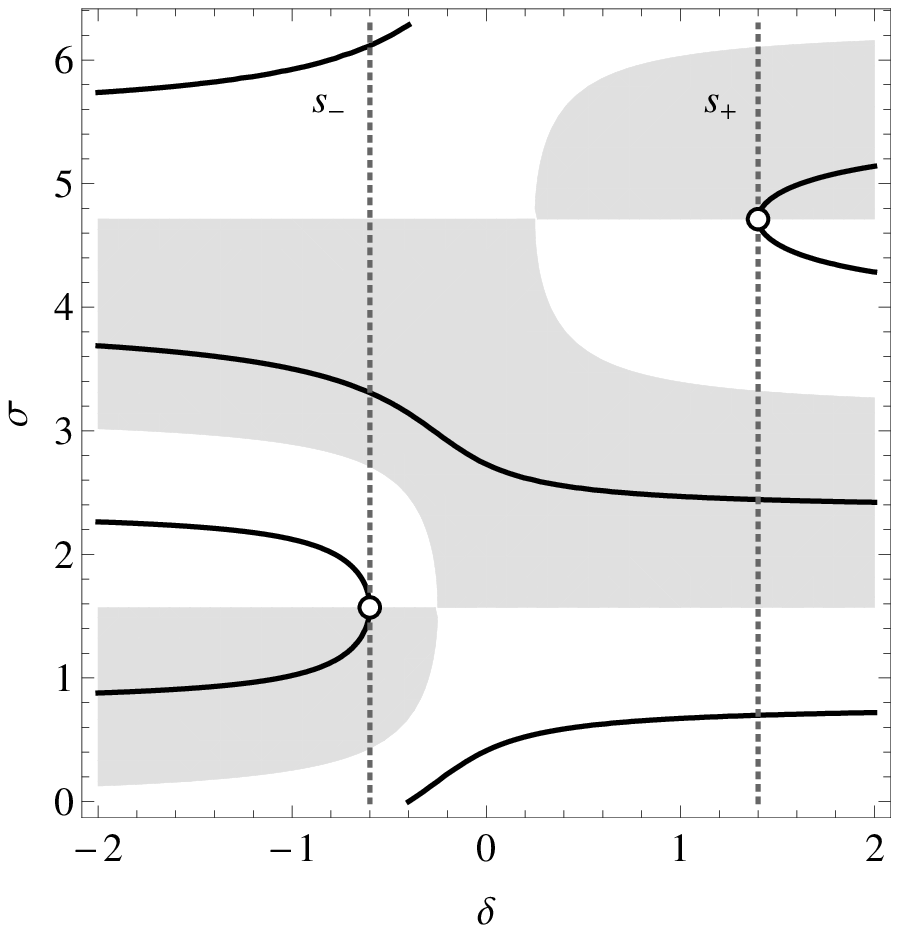}}
\hspace{2ex}
\subfigure[$\kappa=0.9$, $\displaystyle \nu=\frac{\pi}{2}$]{\includegraphics[width=0.3\linewidth]{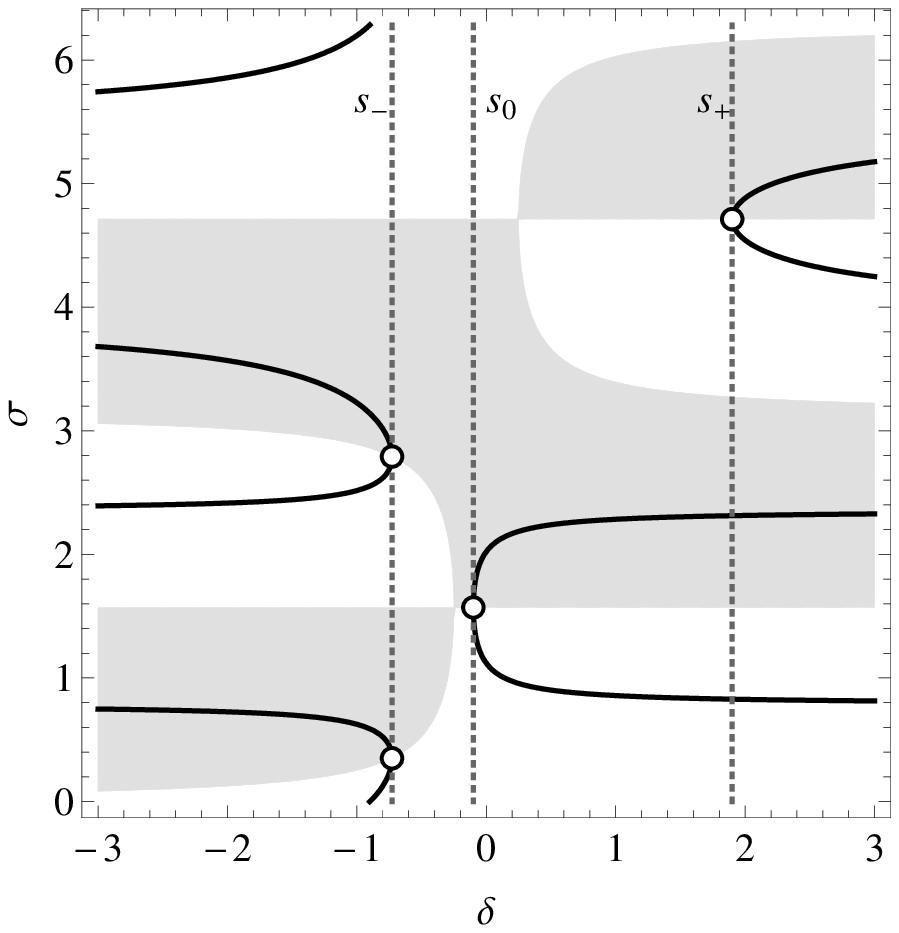}}
\hspace{2ex}
\subfigure[$\kappa=1.6$, $\displaystyle \nu=\frac{\pi}{4}$]{\includegraphics[width=0.3\linewidth]{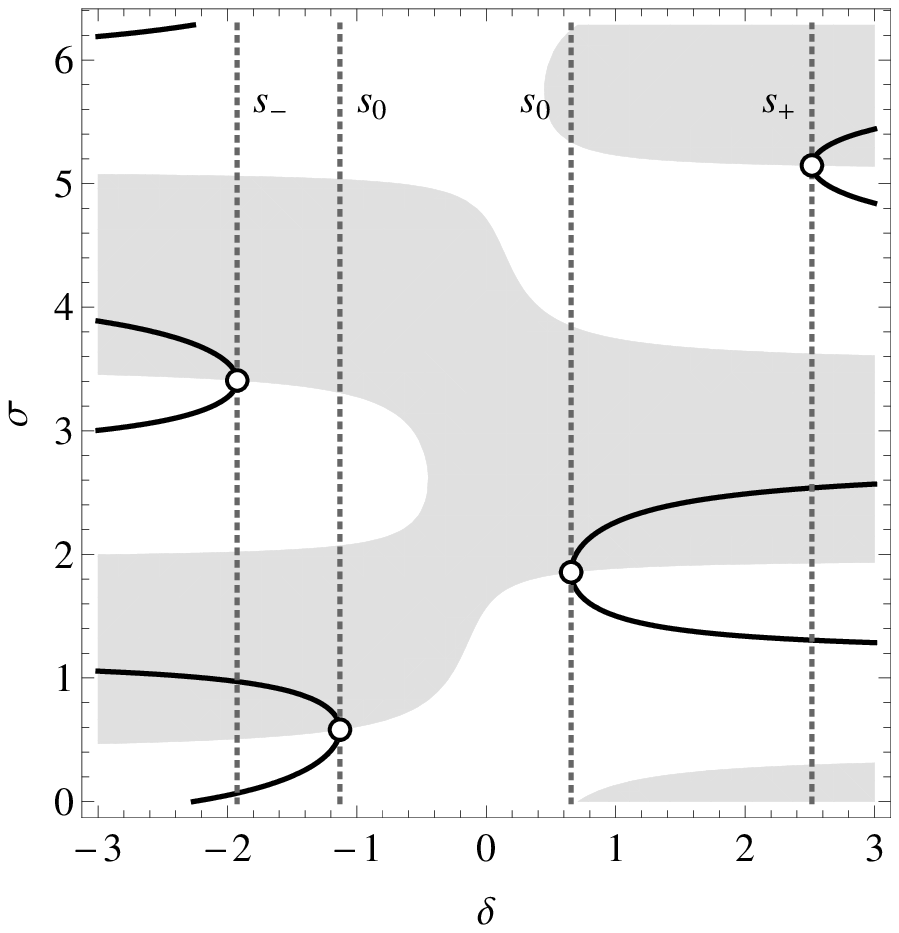}}
\caption{\small The roots to equation \eqref{teq} as functions of the parameter $\delta$ (black solid lines). The vertical dotted lines correspond to $s_-$, $s_0$ and $ s_+$. The shaded areas correspond to $\mathcal P'(\sigma;\delta,\nu,\kappa)>0$, where the particular solutions to system \eqref{MS} with asymptotics \eqref{PAS} are stable.} \label{fig33}
\end{figure}

Some of the autoresonant modes coalesce, when the parameters $(\delta,\nu)$ pass through the bifurcation curves. The following cases are possible. (I) Equation \eqref{teq} has three different roots: two simple roots and one root of multiplicity 2 (see Fig.~\ref{fig34}, a). (II) Equation \eqref{teq} has two different roots: one simple root and one root of multiplicity 3 (see Fig.~\ref{fig35}, a), or two  roots of multiplicity 2 (see Fig.~\ref{fig34}, b). (III) Equation \eqref{teq} has only one root: a root of multiplicity 2 (see Fig.~\ref{fig34}, c), or a root of multiplicity 4 (see Fig.~\ref{fig35}, b).

\begin{figure}
\centering
\subfigure[$\kappa=0.4$, $\displaystyle \nu=\frac{\pi}{2}$]{\includegraphics[width=0.3\linewidth]{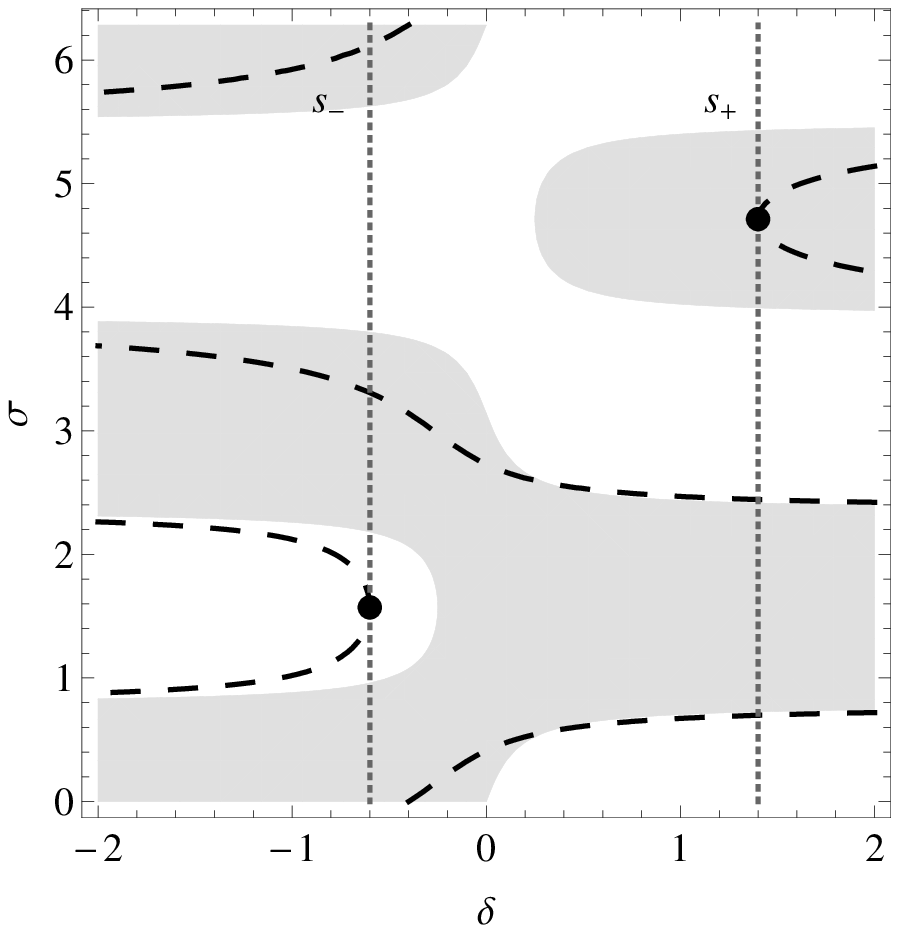}}
\hspace{2ex}
\subfigure[$\kappa=0.9$, $\displaystyle \nu=\frac{\pi}{2}$]{\includegraphics[width=0.3\linewidth]{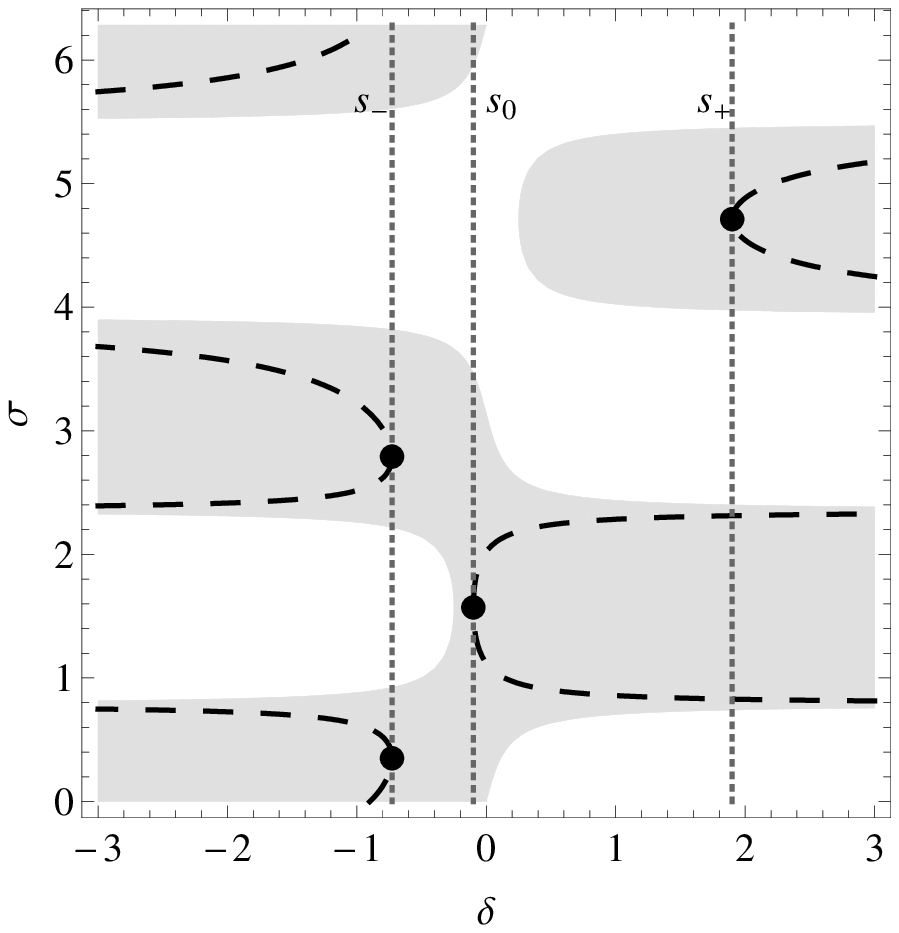}}
\hspace{2ex}
\subfigure[$\kappa=1.6$, $\displaystyle \nu=\frac{\pi}{4}$]{\includegraphics[width=0.3\linewidth]{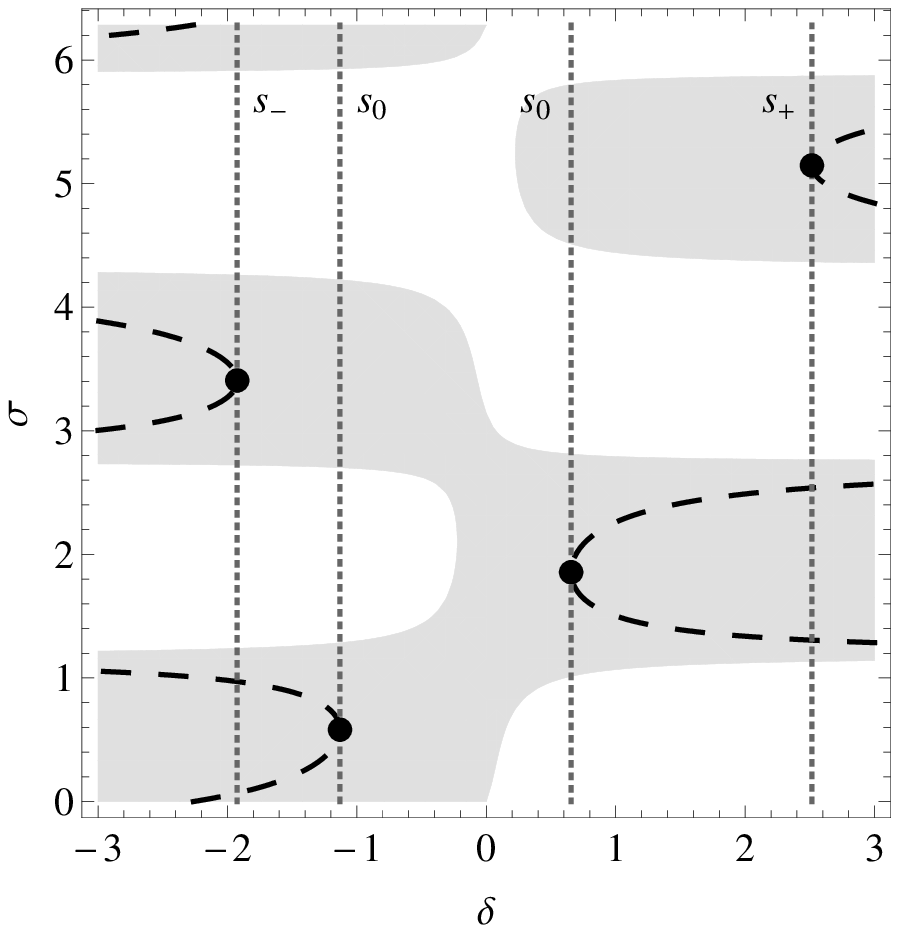}}
\caption{\small The roots to equation \eqref{teq} as functions of the parameter $\delta$ (dashed lines). The vertical dotted lines correspond to $s_-$, $s_0$ and $ s_+$. The black points correspond to  the roots of multiplicity 2. The shaded areas correspond to  $\mathcal P''(\sigma;\delta,\nu,\kappa)>0$, where the particular solution to system \eqref{MS} with asymptotics \eqref{PAS}, $\psi_0=\sigma$, $\psi_1=\phi$ is stable.} \label{fig34}
\end{figure}

\begin{figure}
\centering
\subfigure[$\kappa=0.5$, $\displaystyle \nu\approx 1.038$]{\includegraphics[width=0.3\linewidth]{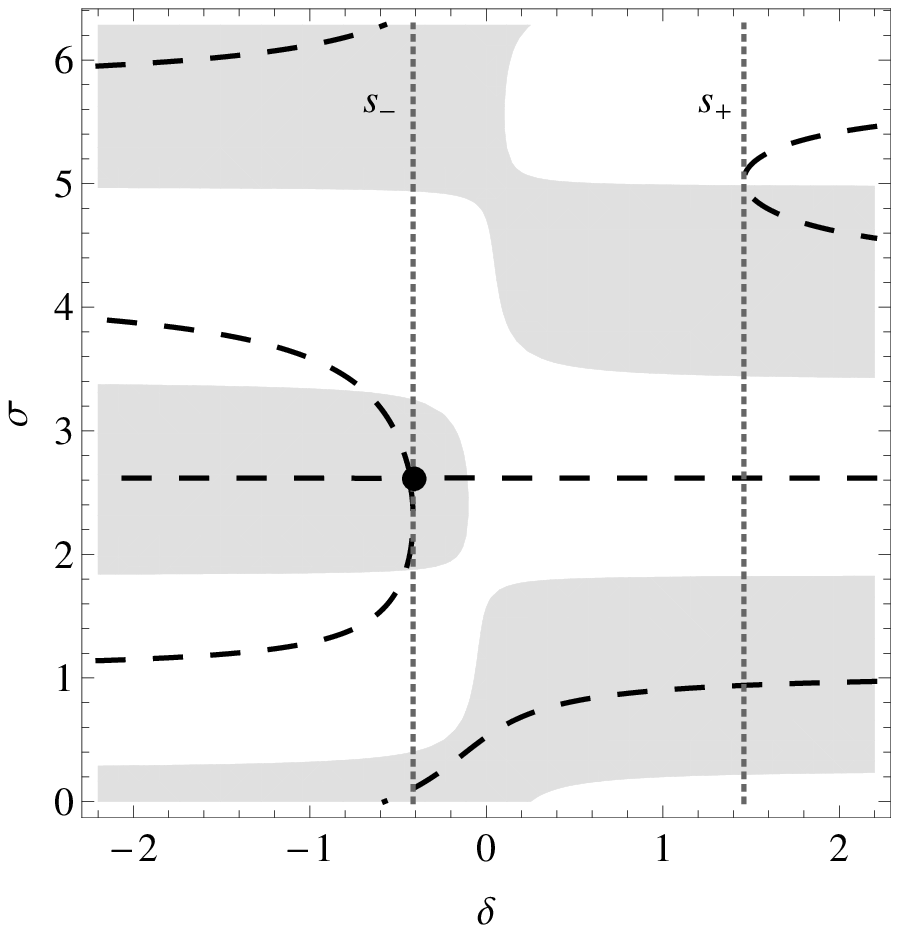}}
\hspace{2ex}
\subfigure[$\kappa=0.75$, $\displaystyle \nu=\frac{\pi}{2}$]{\includegraphics[width=0.3\linewidth]{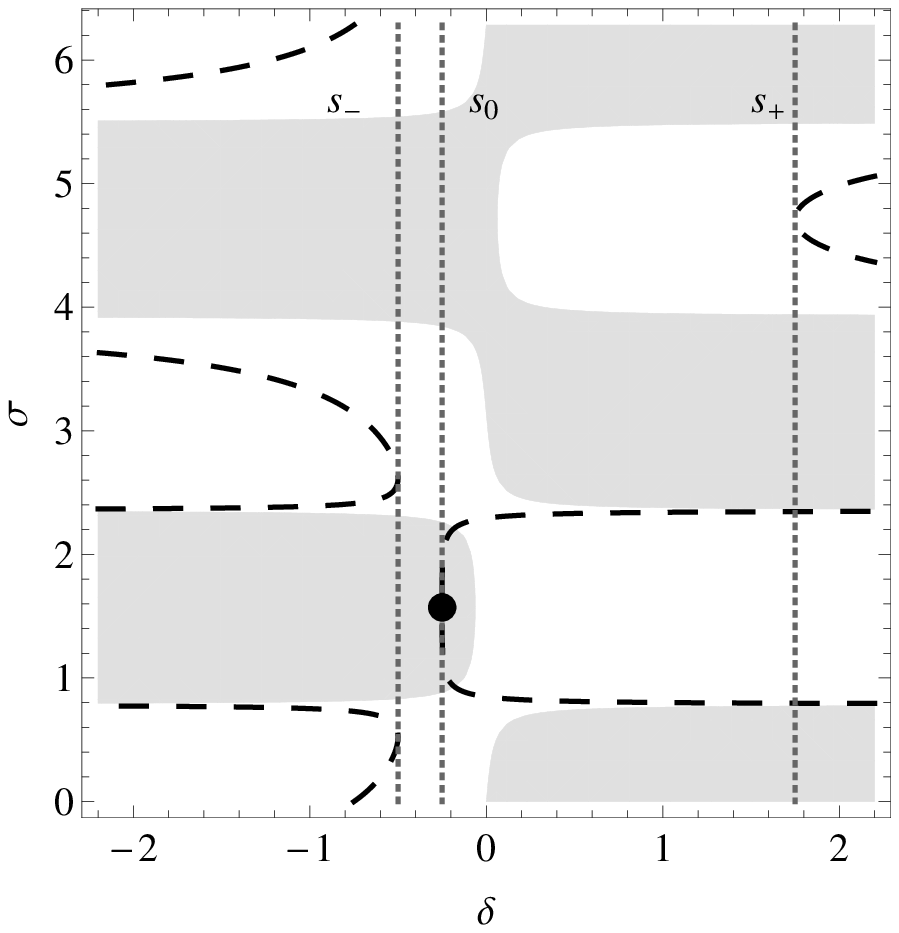}}
\hspace{2ex}
\subfigure[$\kappa=\sqrt{0.3}$, $\displaystyle \nu\approx 2.79 $]{\includegraphics[width=0.3\linewidth]{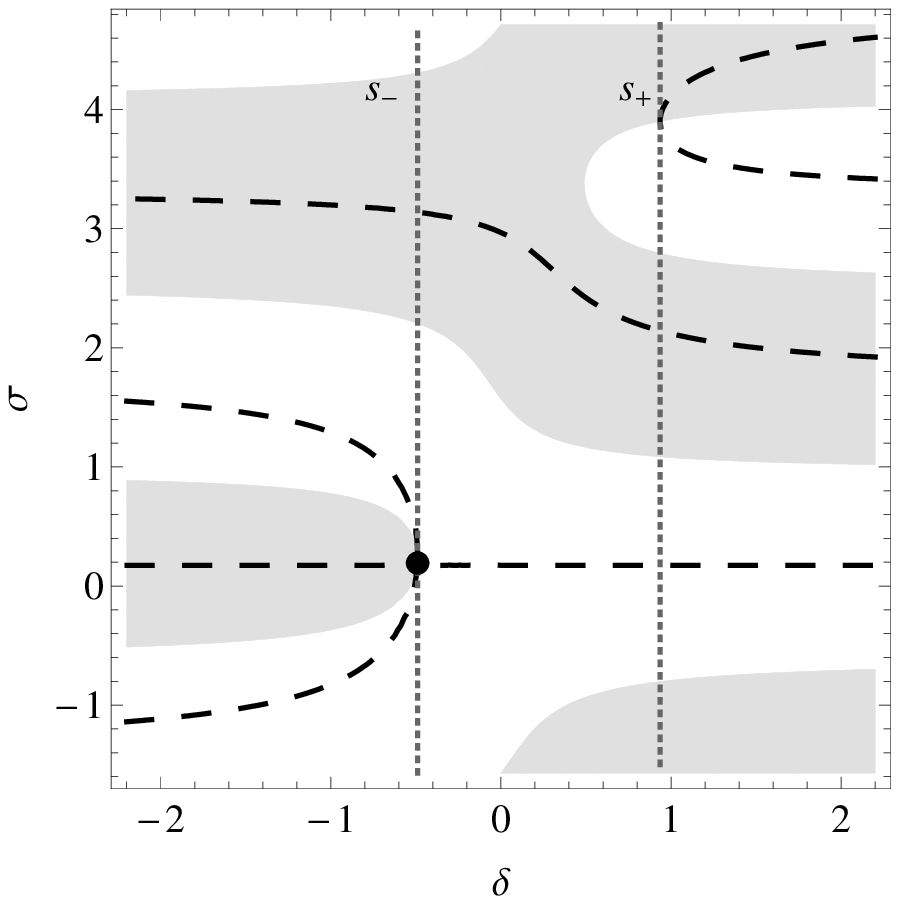}}
\caption{\small The roots to equation \eqref{teq} as functions of the parameter $\delta$ (dashed lines). The vertical dotted lines correspond to $s_-$, $s_0$ and $ s_+$.  The black points correspond to the roots of multiplicity 3 and 4. (a) The shaded areas correspond to $\mathcal P'''(\sigma;\delta,\nu,\kappa)>0$, where the particular solution to system \eqref{MS} with asymptotics \eqref{PAS3} is stable. (b) The shaded areas correspond to  $\mathcal P^{(4)}(\sigma;\delta,\nu,\kappa)>0$. (c) The shaded areas correspond to $\mathcal P'(\sigma;\delta,\nu,\kappa)>0$.} \label{fig35}
\end{figure}

\begin{figure}
\centering{
 \includegraphics[width=0.4\linewidth]{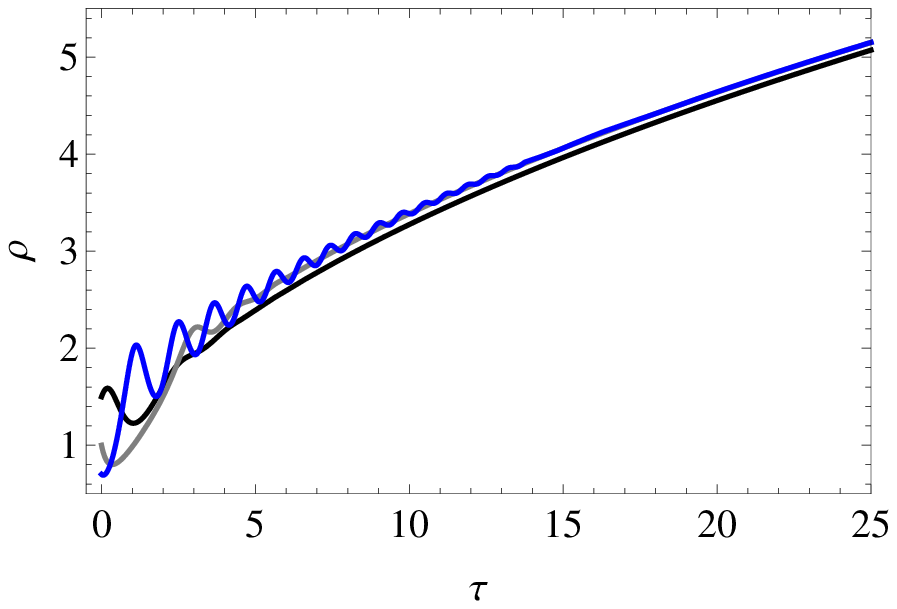}
\hspace{2ex}
 \includegraphics[width=0.4\linewidth]{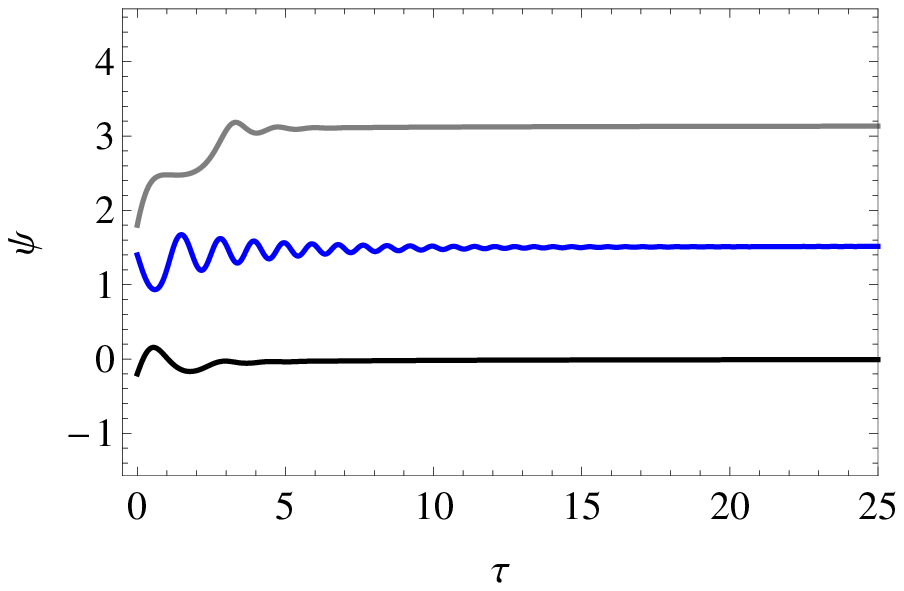}}
\caption{\small  The evolution of $\rho(\tau)$ and $\psi(\tau)$ for solutions of system \eqref{MS} with $\lambda=1$, $\alpha(\tau)\equiv \sqrt\tau$, $\beta(\tau)\equiv \delta$, $\gamma(\tau)\equiv \kappa$. The black curves correspond to $\delta=-2$, $\kappa=1$, $\nu=5\pi/6$. The gray curves correspond to $\delta=-2/\sqrt 3$, $\kappa=1$, $\nu=2\pi/3$. The blue curves correspond to $\delta=-3/2$, $\kappa=1/4$, $\nu=\pi/6$. } \label{rhopsi}
\end{figure}

If $\sigma$ is the root of multiplicity 2, system \eqref{MS} has two autoresonant modes associated with the particular solutions having asymptotics \eqref{PAS}, where $\psi_0=\sigma$ and $\psi_1=\pm\phi$. In this case, the solution with $\psi_1=\pm\phi$ is stable if $\pm\mathcal P''(\sigma;\delta,\nu,\kappa)>0$ (see Fig.~\ref{fig34}).
If $\sigma$ is the root of multiplicity 3, there is a mode corresponding to the particular solution with asymptotics \eqref{PAS3}, where $\psi_0=\sigma$,  $\psi_1=0$ and $\psi_2=\chi$. This mode is stable if $\mathcal P'''(\sigma;\delta,\nu,\kappa)>0$ (see Fig.~\ref{fig35}, a).
Finally, if $\sigma$ is the root of multiplicity 4, there are two autoresonant modes with asymptotics \eqref{PAS4}, where $\psi_0=\sigma$,  $\psi_1=\pm\xi$. In this case, the mode with $\psi_1=\xi$ is exponentially stable, while the mode with $\psi_1=-\xi$ is unstable (see Fig.~\ref{fig35}, b).

Note that the combined excitation allows to expand the use of autoresonant method for control the dynamics of nonlinear systems.  In particular, unstable autoresonant modes in systems with pure external excitation $(\delta=0)$ can be stabilized by switching on parametric pumping (see, for example, Fig.~\ref{fig35}, c, where the mode with $\sigma=\arcsin(4\kappa/3)$ becomes stable as $\delta<-\sqrt{0.24}$).
Moreover, for every $\gamma_0>0$, the parameters $\lambda>0$, $\beta_0\neq 0$ and $\nu\in [0,\pi)$ of the combined excitation can be chosen in such a way to guarantee the existence and stability of autoresonant mode with any prescribed phase shift $\psi(\tau)\approx \sigma$, $\sigma \in [0,2\pi)$. For example, for $\sigma=0$, we should take $\nu=\pi-\arcsin(-\kappa\delta^{-1})$ and $\delta< -\sqrt{\kappa^2+1/4}$. In this case, $\mathcal P(\sigma;\delta,\nu,\kappa)=0$, $\mathcal P'(\sigma;\delta,\nu,\kappa)>0$ and Theorem~\ref{cgs} is applicable. Similarly, for $\sigma=\pi$, we should take $\nu=\pi-\arcsin(-\kappa\delta^{-1})$ and $\delta< -\kappa$. For $\sigma=\pi/2$, one can take $(\delta,\nu,\kappa)$ such that $0< (\kappa-1)\delta^{-1}\leq 1$,  $\nu=\arcsin((\kappa-1)\delta^{-1})$ and $\delta \cos\nu<0$ (see Fig.~\ref{rhopsi}).

\section*{Acknowledgements}
The research presented in Section~\ref{sec3} is funded in the framework of executing the development program of Scientific Educational Mathematical Center of Privolzhsky Federal Area, additional agreement no. 075-02-2020-1421/1 to agreement no. 075-02-2020-1421.

\end{document}